\documentclass[%
 reprint,
 amsmath,amssymb,
 aps, prx,
floatfix,
]{revtex4-2}
\usepackage[table,xcdraw]{xcolor}
\usepackage{graphicx}
\usepackage{dcolumn}
\usepackage{bm}
\usepackage{float}

\usepackage{algpseudocode}
\usepackage{algorithm}
\usepackage{physics}
\usepackage{color}
\usepackage{placeins}
\usepackage[colorlinks=true, linkcolor=blue, citecolor=blue, urlcolor=blue]{hyperref}

\usepackage[capitalise]{cleveref}
\usepackage[all]{hypcap}
\usepackage{booktabs} 
\usepackage{amsmath} 						
\usepackage{amssymb} 						
\usepackage{mathtools}			
\usepackage{mathrsfs}                       
\usepackage{amsthm}
\usepackage{amssymb}
\usepackage{hyperref}
\hypersetup{
    colorlinks = true,
    linkcolor = blue,
    allcolors=blue
    }

\newtheorem{theorem}{Theorem}

\newtheorem{lemma}{Lemma}

\usepackage{algpseudocode}
\usepackage{subcaption}
\extrafloats{100}

\begin{document}

\preprint{APS/123-QED}

\title{Engineering Quantum Reservoirs through Krylov Complexity, Expressivity and Observability}

\author{Saud \v{C}indrak}
\email{saud.cindrak@tu-ilmenau.de}
\affiliation{%
 Technische Universität
Ilmenau, Institute of Physics, Ilmenau, Germany
}%
\author{Lina Jaurigue}
\affiliation{%
 Technische Universität
Ilmenau, Institute of Physics, Ilmenau, Germany
}%

\author{Kathy L{\"u}dge}
 \affiliation{%
 Technische Universität
Ilmenau, Institute of Physics, Ilmenau, Germany
}%
\date{\today}

\begin{abstract}
This study employs Krylov-based information measures to understand task performance in quantum reservoir computing, a sub-field of quantum machine learning.  In our study we show that fidelity and spread complexity can only explain the task performance for short time evolutions of the quantum systems. We then discuss two measures, Krylov expressivity and Krylov observability, and compare them to task performance and the information processing capacity. Our results show that Krylov observability exhibits almost identical behavior to information processing capacity, while being three orders of times faster to compute. In the case when the system is undersampled Krylov observability best captures the behavior of the task performance.
\end{abstract}

\maketitle

\section{Introduction}\label{Introduction}
Quantum machine learning (QML) models attempt to exploit the exponential scaling of Hilbert space for computation. A subfield of QML that has garnered significant attention in physics is quantum reservoir computing (QRC). QRC is inspired by classical reservoir computing, where physical systems are used to solve machine learning tasks \cite{JAE01, MAA02, VER07, APP11, MUE24, JAU24, JAU25}. QRC research, particularly with the Ising model, has been explored in \cite{FUJ17, NAK19, KUT20, MAR21, XIA22, GOE23, NOK24, ABB24a}, with initial implementations on IBM's quantum processor \cite{CHE20b, SUZ22, PFE22, YAS23}. It has been shown that quantum reservoir computing can utilize noise and dissipation as resources \cite{SAN24, FRY23, DOM23}. To address the time complexity problem for time-series tasks, weak measurements \cite{MUJ23, FRA24} and reinitialization schemes \cite{CIN24} have been proposed, reducing complexity and improving performance. Other approaches include reintroducing measured outputs \cite{KOB24, AHM24} and continuous measurements \cite{ZHU24a}, while \cite{SCH25} utilizes techniques from parameterized quantum circuits to define expressive limits in QRC. 
However, understanding in quantum machine learning and quantum reservoir computing remains limited. In classical machine learning, this understanding arises from expressivity and explainability measures, which is still an active field and is mostly unresolved in quantum machine learning. 
A key challenge is comparing quantum systems described by unitaries or Hamiltonians based on their expressivity within the Hilbert space. An ideal information measure would provide insights into task performance while offering a physical interpretation of the system, making it physically interpretable.

In this work we discuss various measures regarding task performance in quantum reservoir computing. We train four distinct quantum reservoirs on a chaotic time series and larger reservoirs with random inter-spin couplings. As a benchmark the Lorenz63 system is used \cite{LOR63}. We compare fidelity, information processing capacity \cite{DAM12}, and spread complexity \cite{BAL22, PAR19} as established methods. Additionally, we discuss Krylov expressivity and Krylov observability, two measures designed to capture the effective phase space dimension of the Hilbert space\cite{CIN25, CIN24a}. We demonstrate that Krylov observability and information processing capacity perform best in explaining the behavior of task performance. 

Spread and operator complexity have successfully captured various effects of time evolution in relation to spin chains, SYK models, chaotic quantum systems, driven quantum systems, and others \cite{ALI23, BHA22a,  PAR19, AFR23, ANE24, BAE22, BAL22, BAR19, BHA24a, CAM23, CAO21, CAP21, CHA23, DYM21, FAN22, GUO22A, HAS23, HE22, HEV22, IIZ23, JIA21, KIM22, LI24, LIU23b, MAG20, MUC22, NIZ23, PAT22, RAB22, RAB22a, VAS24, CAP24, BAL23, CRA23, AGU24, CAP22, CAP22a, ERD23, GIL23, NAN24, PAL23, BHA22, CHA23, GAU23, NIE06a, BHA24, BAL23a, LIN22, RAB23, HUH24, ZHO25b, CRA25, SUC25}. A detailed explanation on the methods is given in \cite{NAN24a}. However, a major challenge in calculating spread and operator complexity is the necessity of classical simulation for constructing the Krylov spaces. In \cite{CIN24a}, we address this by introducing quantum-mechanically measurable Krylov spaces for the computation of spread complexity and propose \textit{Krylov expressivity}, a measurable expressivity measure of the effective phase space dimension. In \cite{CIN25}, we extend this research by proposing \textit{Krylov observability}, a measure to quantify the effective phase space dimension of various observables measured multiple times. Our results show that Krylov observability effectively captures how well a quantum system can retain and map macroscopic data non-linearly, thus highlighting the importance of Krylov spaces in quantum reservoir computing.

The authors in Ref. \cite{DOM24} first explored the insights that spread complexity can provide in the field of quantum reservoir computing. The authors found that quantum reservoirs with larger mean spread complexity tended to perform better.

In \cite{VET25} the authors discuss how state estimation in a quantum extreme learning machine can occur beyond the scrambling time, which challenges the common belief that information cannot be retrieved after this point. A similar pattern is observed in quantum reservoir computing, where task performance initially increases and then saturates \cite{MAR20}. We show that fidelity and spread complexity can only explain the initial change in task performance, but fails to explain the saturation in task performance. Our work explains this saturation using the concepts of Krylov expressivity and Krylov observability. The initial state or data is mapped onto the Krylov space, where expressivity and observability increase before reaching a saturation point. Beyond this point, further increases in time do not enhance the system's expressivity, resulting in stable task performance.

Among all measures, Krylov observability and information processing capacity explain the general trend of task performance best and show almost identical behavior for larger readout dimensions. For small readout dimensions, we show that Krylov observability performs better than information processing capacity in capturing the behavior of the Lorenz task. This is due to the fact that information processing capacity is upper-bounded by the readout dimension, while Krylov observability adapts better to a smaller number of measurements.
Our results suggest that Krylov expressivity is not as insightful as Krylov observability in quantum reservoir computing. This is because Krylov expressivity only considers how input states are mapped onto the Krylov space. However, if the initial state is sampled from an encoding unitary, Krylov expressivity can further the understanding of different encoding strategies. Two sub-fields of quantum machine learning with various encoding strategies are variational quantum machine learning and quantum extreme learning \cite{SCH18a, ROM17, DUN18, KIL19, SCH14, FAR18, BRU13a, LI15, LI17, SAG21, SCH21, LIU24, HUA04, DIN15, WAN22, HUA15, HUA06, HUA12, MUJ21, INN23, GHO19, GHO19a, XIO23, SUP24, QI24}. Krylov observability, on the other hand, can provide insights into how much of the Krylov space of a variational quantum algorithm can be sampled. Current research on the explainability of quantum machine learning is sparse and primarily employs techniques from classical machine learning in QML \cite{STE22, PIR24, LIF22, HEE23}, while \cite{GIL24} explores the possibilities of designing explanation techniques for parameterized quantum circuits.
We further give a brief discussion of the number of matrix operations required for the computation of Krylov observability ($N_{\mathrm{obs}}$) and for the computation of the state matrix ($N_{\mathbf{u}}$) , which is required for the IPC. The ratio of these two scales with $r={N_\mathrm{obs}}/{N_\mathrm{sys}}\in \mathcal{O}(N_R/N_{\mathbf{u}})$, where $N_{\mathbf{u}}$ is the length of the time-series and $N_R$ is the readout dimension. In typical RC approaches $N_u\gg N_R $ is given, which results in $r\ll 1$ and in our experiments in $r=0.00075$. Similar to the discussion about the grade of the Krylov state space presented in \cite{CIN24a}, this work introduces a similar identity for Krylov operator spaces, where the grade of the Krylov operator space can be derived using only the spectral properties of the Hamiltonian and the operator. This is of interest because numerical errors during the orthonormalization can artificially increase the apparent dimension of the Krylov space. Knowing the Krylov grade in advance allows for a cutoff based on physical rather than numerical considerations. This understanding is then used to explain why the operator Krylov space can be smaller for certain Hamiltonians, even when their corresponding Krylov state spaces are larger than those of others. 

This work is organized as follows. \cref{sec:QRC} introduces quantum reservoir computing, information processing capacity, and the Lorenz task. \cref{sec:Krylov} discusses Krylov spaces in quantum mechanics. Here, we discuss that any time-evolved state lies within a Krylov space and explain measurable Krylov spaces, leading to the introduction of Krylov expressivity. This concept is extended to operator complexity, along with a brief discussion of Krylov observability as proposed in \cite{CIN25}. 
In \cref{sec:Results}, we demonstrate the limitations of fidelity and spread complexity in explaining task performance due to their oscillatory behavior, even though Lorenz task performance remains mostly constant. We show that Hamiltonians with higher expressivity can exhibit worse performance when Krylov observability is lower. We conclude our work with a discussion of the measures presented.

\section{Quantum Reservoir Computing}\label{sec:QRC}
\begin{figure*}[t]
    \hspace*{-0.5 cm}
    \centering
    \includegraphics[scale=0.9]{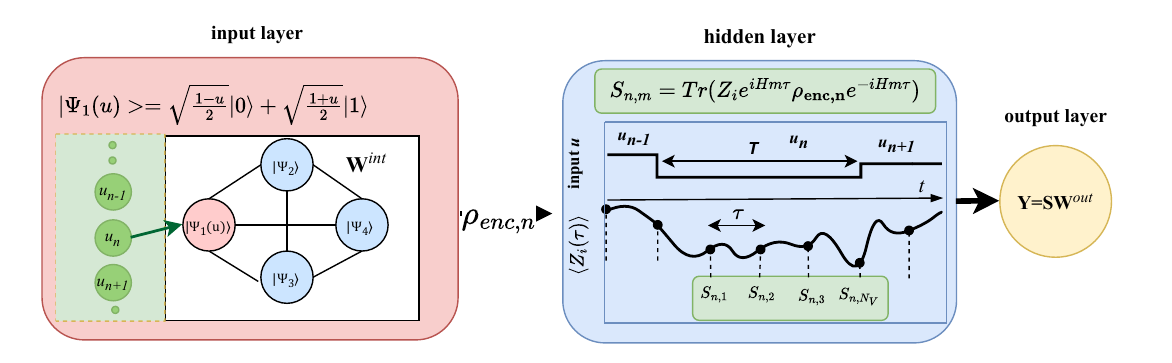}
    \caption[Reservoir Computing scheme]{A time series $u_n$ is encoded into a qubit of a quantum system, described by a Hamiltonian $H$. The system state is described by the density matrix $\rho_{enc}$ evolving over time $t$. Measurements of some observables $Z_i$ at times $m\tau$ are taken to construct the state matrix $\mathbf{S}$. The total number of measurements per observable is $V$ and the clock cycle is defined as $V\tau=T$. In the output layer, the state matrix $\mathbf{S}$ is multiplied by the readout weights $\mathbf{W}^{out}$ to construct the output vector $\mathbf{Y}=\mathbf{S}\mathbf{W}^{out}$.}
    \label{fig:1_QRC}
\end{figure*}
Quantum reservoir computing makes use of a quantum system as a reservoir for classical or quantum tasks. \cref{fig:1_QRC} shows a sketch of a quantum reservoir scheme and is explained in the following.
\begin{enumerate}
    \item Initialization: The reservoir is first initialized with an initial time series, in our case of length $N_{\mathrm{In}} = 10,000$
    \item Input layer (green) and encoding (red): The encoding of the $n$-th input $u_n\in \{u_1,..,u_{N_u}\}$ into the first qubit leads to the encoded state 
    \begin{align}
        \ket{\Psi_n} = \sqrt{\frac{1-u_n}{2}}\ket{0}+\sqrt{\frac{1+u_n}{2}}\ket{1}.
    \end{align}
    After encoding the state of the system is described by
    \begin{align}
        \rho_{enc,n}= \ket{\Psi_n}\bra{\Psi_n}\otimes \Tr_1(\rho_{n}).
    \end{align}
     The overwriting of the first qubit at each time-step results in loss of information about previous time steps. This is called fading memory property and is a crucial requirement for reservoir computing, as reservoirs exhibit a memory-nonlinearity trade-off \cite{JAE01, MAA02, VER10, BUT13b, INU17}. In \cite{CIN24} this memory-nonlinear trade-off was shown in quantum reservoir computing.
    \item Hidden Layer (blue): The reservoir evolves for a clock cycle $T$ by the unitary evolution $$U_R = \exp (-iHT).$$
    A discrete set of observables $\{Z_k\}_{k=1..,K}$ is measured to construct the expectation values. The expectation values at time $t_j$ after the $n$-th input are given by
    \begin{align}
        \ev{Z_k(t_j,u_n)} = Tr(Z_ke^{-iHt_j}\rho_{enc,n}e^{iHt_j}).
    \end{align}
    A sketch of the evolution of an observable with respect to different inputs is shown above. To increase the readout dimension, each observable is sampled \( V \) times for each input \( u_n \) at times \( t_j = nT + (j+1)\tau \), with \( \tau = T/V \). The sketch \cref{fig:1_QRC} shows a quantum reservoir with four virtual nodes $V=4$. This results in $N_R=VN_S$ readout nodes for each input $u_n$, where $N_S$ represents the number of observables measured. Writing the readouts for each input into a row, results in the state matrix $\mathbf{S}\in \mathbb{R}^{N_u, N_R}$ .
    
    \item Output layer (yellow): The state matrix is multiplied with readout weights $\mathbf{W}^{out}$ resulting in the output 
    \begin{align}
        \mathbf{Y} = \mathbf{S}\mathbf{W}^{out}.
    \end{align}
    \item Training: For training, the data is separated into a initialization set $u_{\mathrm{In}}$, a training set $u_{\mathrm{Tr}}$ and a testing set $u_{\mathrm{Te}}$ of lengths $N_{\mathrm{In}}$, $N_{\mathrm{Tr}}$  and $N_{\mathrm{Te}}$ respectively. A buffer of length $N_{b}$ is added between the sets.
    We first initialize the system with $N_{\mathrm{In}}$ steps, afterwards the weights are optimized in regards to the training set $u_{\mathrm{Tr}}$, resulting in the trained weights $\mathbf{W}^{out}$. The testing set is used to check if the reservoir can successfully generalize to new data. The weights are optimized to minimize the loss 
    \begin{align}
        L &:= (\mathbf{Y}-\mathbf{Y}^{targ})^2,
    \end{align}
    where $\mathbf{Y}^{targ}$ is the target vector. To consider shot noise of the measurement, Gaussian distributed noise $\mathcal{N} \in \mathbb{R}^{N_u, N_R}$ will be added on the state matrix 
    \begin{align}
        \mathbf{S} &\leftarrow \mathbf{S} + \eta \mathcal{N}.
    \end{align}
   The noise term $\eta\mathcal{N}$ can be used as a regularization parameter when computing the state matrix, this is typically addressed as regularization by noise.
   The readout weights $\mathbf{W}^{out}$ are then computed by
    \begin{align}
        \mathbf{W}^{out}&=(\mathbf{S}_{Tr}^T\mathbf{S}_{Tr})^{-1}\mathbf{S}_{Tr}^T\mathbf{Y}^{targ}.
    \end{align}
\end{enumerate} 

\textbf{Transverse field Ising Hamiltonian}\\
Four transverse field Ising Hamiltonian with different inter-spin couplings will be used as quantum reservoirs and analyzed in regards to spread complexity, fidelity, Krylov observability, Krylov expressivity and task performance.
\begin{align}
{H}_{I\alpha} = \sum_{i=1,j>i}^{N_S}J_{ij}X_iX_j+\sum_{i=1}^{N_S}hZ_i,
\label{eq:Ising_Ham}
\end{align}
where $\alpha\in \{1,2,3,4\}$ is the index referencing each of the four Hamiltonian, $N_S$ is the number of qubits, $X_i, Y_i$ and $Z_i$ are the $\sigma_x$, $\sigma_y$ and $\sigma_z$ Pauli matrices of the $i$th qubit given by
\begin{align}
    X_i, Y_i, Z_i = \Big( \bigotimes_{k=1}^{i-1}I_2 \Big)\otimes \sigma_{x,y,z} \otimes\Big( \bigotimes_{k=i+1}^{N_S}I_2 \Big),
\end{align} 

and $h=0.5$ is set for all Hamiltonians. The number of pairwise distinct eigenvalues $d$ and the inter-spin couplings $J_{i,j}$ are shown in \cref{table:couplings}.\\
\hfill \break

\begin{table}[H]
\centering
\begin{tabular}{lcccc}
\toprule
 & ~$H_{I1}~$ & $~H_{I2}~$ & $~H_{I3}~$ & $~H_{I4}~$ \\
\midrule
\# of eigenvalues $d~~~$ & 9 & 16 & 15 & 16 \\
$J_{1,2}$             &  0.50 & 0.40 & 0.35 & 0.35 \\
$J_{1,3}$             &  0.50 & 0.50 & 0.40 & 0.40 \\
$J_{1,4}$             &  0.50 & 0.50 & 0.45 & 0.45 \\
$J_{2,3}$             &  0.50 & 0.50 & 0.50 & 0.50 \\
$J_{2,4}$             &  0.50 & 0.50 & 0.55 & 0.55 \\
$J_{3,4}$             &  0.50 & 0.50 & 0.60 & 0.65 \\
\bottomrule
\end{tabular}
\caption{The number of pairwise distinct eigenvalues $d$ and the inter-spin couplings $J_{i,j}$ for the different Ising reservoirs are shown.}
\label{table:couplings}
\end{table}

\textbf{Lorenz Task}\\
The Lorenz63 system is a commonly used to construct time series prediction benchmark tasks \cite{LOR63}. The dynamics of the system are governed by \cref{eq:Lor}
\begin{align}
    \dot{x} &= a(y-x) \nonumber \\
    \dot{y} &= x(b-z)-y \nonumber \\
    \dot{z} &= xy - cz.
    \label{eq:Lor}
\end{align}
With the parameters set to $a=10$, $b=28$, $c=8/3$, the Lorenz63 system exhibits chaotic dynamics. We simulate Eqs.~\eqref{eq:Lor} with an integration time step $dt = 0.001$. To construct the input and target data for the task, we down sample to a  discretized step of $\Delta t=0.02$, which results in the time series $x_n=x(n\Delta t)$, $y_n=y(n\Delta t)$, and $z_n=z(n\Delta t)$. We consider two time series prediction tasks. The first considers the time series $x_n$ as input and tries to predict the $p$th future step $x_{n+p}$, referred to as the LXX task. In our study we choose to predict $\Delta t=0.1$ into the future, which results in a $p=5$ step ahead prediction. The second task takes $x_n$ as input but tries to predict $z_{n}$. This task is referred to as the LXZ task or cross-prediction task.
\\\\
\textbf{Error Measure}\\
The normalized root mean squared error (NRMSE) between output $Y$ of the reservoir and target $Y^{targ}$ is defined as 
\begin{align}
    \mathrm{NRMSE}= \sqrt{\frac{(\mathbf{Y}-\mathbf{Y}^{\mathrm{targ}})^2}{N\mathrm{var}(\mathbf{Y}^{\mathrm{targ}})}} = \sqrt{1-
C(\mathbf{Y}, \mathbf{Y}^{\mathrm{targ}})}
\end{align}
where the capacity or Pearson correlation is defined as 
\begin{align}
&C(\textbf{Y}, \textbf{Y}^{\mathrm{targ}}) = \frac{\mathrm{cov}(\textbf{Y}, \textbf{Y}^{\mathrm{targ}})}{\sigma^2(\textbf{Y})\sigma^2(\textbf{Y}^{\mathrm{targ}})}.
\label{eq:Pearson}
\end{align}
 $\textrm{cov}(\cdot)$ and $\sigma (\cdot)$ are the covariance and standard deviation, respectively.
\\\\

\textbf{Information processing capacity}\\
The information processing capacity (IPC) is a measure of how well a reservoir can generalize prior data onto a set of orthogonal functions \cite{DAM12}. In this study, Legendre polynomials are used as target functions, which requires the input series to be sampled from a uniform distribution, i.e., $u_n \in \mathcal{U}([-1,1])$\cite{KUB21}. The IPC does not predict task performance, but is a measure to gain understanding of the reservoir and the underlying behavior of reservoir computers. \cite{CIN24} showed that the memory–nonlinearity trade-off can be used to increase the nonlinear term of the IPC, which is proved to be beneficial for chaotic time-series prediction tasks. 

The first-order IPC, denoted \( \mathrm{IPC}_1 \), is also known as linear IPC, memory capacity, or linear short-term memory. It quantifies how well the system can recall previous data \cite{JAE02}. Let $\mathbf{u} = (u_{n_1}, u_{n_1+1}, \ldots, u_{n_2})$ with $n_2 > n_1$. The target is to reconstruct the time series $m$ steps into the past:
\[
\mathbf{Y}^{\mathrm{targ}}(-m) = (u_{n_1-m}, u_{n_1+1-m}, \ldots, u_{n_2-m}).
\]
The first-order IPC is then defined as
\begin{align}
    \mathrm{IPC}_1 = \sum_{m} C(\mathbf{Y}, \mathbf{Y}^{\mathrm{targ}}(-m)),
\end{align}
where $\mathbf{Y}$ is the output of the reservoir.

In practice, the target is given by
\[
\mathbf{Y}^{\mathrm{targ}}(-m) = (l_1(u_{n_1-m}), l_1(u_{n_1+1-m}), \ldots, l_1(u_{n_2-m})),
\]
where $l_1$ is the first-order Legendre polynomial, defined as $l_1(x) = x$. Therefore, the first-order $\mathrm{IPC}_1$ is equivalent to the memory capacity. The Legendre polynomials are defined as:
\begin{align}
    l_0(x) &= 1, \quad l_1(x) = x, \quad l_2(x) = \frac{1}{2}(3x^2 - 1) \nonumber \\
    l_n(x) &= \frac{2n + 1}{n + 1} x l_n(x) - \frac{n}{n + 1} l_{n - 1}(x), \quad \text{for } n \geq 2
\end{align}

For readability, we define:
\[
\mathbf{l}_k(-m) = (l_k(u_{n_1-m}), l_k(u_{n_1+1-m}), \ldots, l_k(u_{n_2-m})).
\]

The IPC can be generalized to higher orders. For second order, $\mathrm{IPC}_2$ is computed by:
\begin{align}
    \mathrm{IPC}_{2,1} &= \sum_{m} C(\mathbf{Y}, \mathbf{l}_2(-m)) \nonumber \\
    \mathrm{IPC}_{2,2} &= \sum_{\substack{m_1, m_2 \\ m_1 > m_2}} C(\mathbf{Y}, \mathbf{l}_1(-m_1) \mathbf{l}_1(-m_2)) \nonumber \\
    \mathrm{IPC}_2 &= \mathrm{IPC}_{2,1} + \mathrm{IPC}_{2,2}.
\end{align}
We note that there are two second-order contributions. In the case of $\mathrm{IPC}_{2,2}$, the constraint $m_1 > m_2$ is necessary because the case $m_1 = m_2$ is already covered by $\mathrm{IPC}_{2,1}$, and because the product $l_1(m_1) l_1(m_2)$ is commutative, making $m_1 > m_2$ sufficient to represent all unique combinations.

This idea can be extended to higher-order IPCs. The third-order IPC consists of three types of hyper-tasks, involving:
\[
l_3(m), \quad l_1(m_1) l_2(m_2), \quad \text{and} \quad l_1(m_1) l_1(m_2) l_1(m_3).
\]
We see that third-order IPC involves three indices $m_1, m_2$, and $m_3$, and all valid combinations must be considered. Each combination requires training the reservoir.
If only the past fifteen inputs are considered, there are already $\binom{15}{3} = 455$ distinct target combinations for $m_1, m_2$, and $m_3$. This highlights the computational cost of IPC, as it requires training the reservoir thousands of times if even higher IPCs have to be computed.
Higher-order $\mathrm{IPC}_i$ are defined analogously, as discussed in detail in \cite{JAE01, CIN24}. The total IPC is given by the sum of all individual orders:
\begin{align}
    \mathrm{IPC} = \sum_{i} \mathrm{IPC}_i.
    \label{eq:IPC}
\end{align}

While $\mathrm{IPC}$ is often referred to as memory capacity in the literature \cite{GOL20, KOE20a, BAU22c}, in this work we use the term memory capacity specifically for the first-order $\mathrm{IPC}_1$.
One important identity is that the IPC is upper bounded by the readout dimension $N_R$, i.e. $IPC\leq N_R$. 
The trends of IPC and task performance often resemble each other when the readout dimension is sufficiently large. 
In \cite{KOE21, KOE22}, analytical evidence is provided showing the influence of eigenvalues on $\mathrm{IPC}$ in classical reservoir computing. Furthermore, it has been demonstrated that $\mathrm{IPC}$ can help predict task performance in reservoir computing \cite{HUE22, HUE22a}. The concept has also been extended to quantum reservoir computing in \cite{MAR20, CIN24}.\\

The computation of the IPC requires simulating the state matrix $\mathbf{S}$. For this reason, we will briefly discuss the number of matrix multiplications needed to construct $\mathbf{S}$. 
Each reservoir evolution requires the following computational steps:
\begin{align}
    \ket{\Phi_n} &= \sqrt{\frac{1 - u_n}{2}} \ket{0} + \sqrt{\frac{1 + u_n}{2}} \ket{1} \nonumber \\
    \rho_n &= e^{-iHt} \left( \ket{\Phi_n}\bra{\Phi_n} \otimes \mathrm{Tr}_1(\rho_{n-1}) \right) e^{iHt} \nonumber \\
    \ev{Z_i(t)} &= \mathrm{Tr}(Z_i \rho_n)
\end{align}
The computation of the trace requires 4 matrix multiplications, with an additional 2 from the exponentials, 1 from the encoded state, and another $K$ for the number of observables measured. This results in a total of $7 + K$ matrix multiplications per evolution step.
Let $N_{\mathbf{u}} = N_{\mathrm{Tr}} + N_{\mathrm{Te}}$, where $N_{\mathrm{Tr}}$ and $N_{\mathrm{Te}}$ represent the number of training and testing steps, respectively, and let $V$ be the number of multiplexing operations. Then the total number of matrix multiplications is given by:
\begin{align}
N_{\mathrm{state}} = (7+K)N_{\mathbf{u}}V
\label{eq:N_state}
\end{align}

\textbf{Simulation parameters}\\
In \cref{table2:para} the simulation parameters for the computation of the results are listed. The information processing  capacity is computed up to the fourth order. 
\begin{table}[H]
\centering
\begin{tabular}{lc}
\toprule
\textbf{Parameter} & \textbf{Value} \\
\midrule
$N_{\mathrm{In}}$ (input length)         & 10{,}000 \\
$N_{\mathrm{Tr}}$ (training length)      & 25{,}000 \\
$N_{\mathrm{Te}}$ (test length)          & 5{,}000  \\
Number of spins $N_S$                    & 4      \\
Regularization $\eta$                   & $10^{-5}$ \\
Buffer $N_b$                             & 100    \\
\bottomrule
\end{tabular}
\caption{QRC and QELM simulation parameters.}
\label{table2:para}
\end{table}

\section{Krylov Spaces in Quantum Mechanics}\label{sec:Krylov}
This chapter introduces Krylov spaces in relation to quantum evolution. We first discuss spread complexity \cite{BAL22}, a measure which quantifies the spread over the Krylov basis. Krylov expressivity is defined as a quantum-mechanically measurable quantity that quantifies the effective phase space dimension of evolved quantum states \cite{CIN24a}. Next, operator complexity is explored, defining a measure of how operators spread over a Krylov operator basis \cite{PAR19}. Building on this, Krylov observability is introduced as a measure that effectively captures the phase-space dimension of a set of operators ${O_1,..,O_K}$, considering the number of measurements involved in constructing the effective Krylov space. In \cite{CIN25}, we demonstrate that Krylov observability effectively captures the behavior of the $\mathrm{IPC} $ and the generalization properties of the quantum reservoir.

\subsection{Krylov Spaces for State Evolution}
The observation of the time-evolution operator being a map onto a Krylov space was first discussed in \cite{PAR19,BAL22}. There, the proof for \cref{eq:krylov1} was given, which we will be going over briefly.
\\
\textbf{Spread Complexity}\\
The Schrödinger equation for the time-independent Hamiltonian $H$ with initial condition $\ket{\Psi(0)}$ is given by:
\begin{align}
    \partial_{t} \ket{\Psi(t)} = -iH\ket{\Psi(t)} \nonumber \\
    \ket{\Psi(0)} := \ket{\Psi_0}.
\end{align}
The solution to this equation is
\begin{align}
    \ket{\Psi(t)} = e^{-iHt}\ket{\Psi_0}=\sum_{k=0}^\infty \frac{(-iHt)^k}{k!} \ket{\Psi_0},
\end{align}
where the series representation of the matrix exponent was used. Defining the linear function $f(\ket{\Psi_0)}:=-iH$ leads to
\begin{align}
    \ket{\Psi(t)} &= e^{-iHt}\ket{\Psi_0} \nonumber \\ 
    &=\sum_{k=0}^\infty (-iH)^k\frac{t^k}{k!}\ket{\Psi_0} =\sum_{k=0}^\infty f^k(\ket{\Psi_0})\frac{t^k}{k!}.
\end{align}
For any time $t\in \mathbb{R}$ it holds that
\begin{align*}
    \ket{\Psi(t)} \in \mathrm{Span}\{f^0(\ket{\Psi_0}), f^1(\ket{\Psi_0})t, f^2(\ket{\Psi_0})\frac{t^2}{2!}, \ldots\}
\end{align*}
Since $t^k/k!>0$ for $t>0$ it follows that 
\begin{align*}
    \ket{\Psi(t)} \in \mathrm{Span}\{f^0(\ket{\Psi_0}), f^1(\ket{\Psi_0}),\ldots\} := K_{\infty}.
\end{align*}
$f$ being a linear function and considering the Krylov space property implies that there exists a $m\leq N$, such that 
\begin{align}
    K_{\infty} &= \mathrm{Span}\{f^0(\ket{\Psi_0}), f^1(\ket{\Psi_0}), f^2(\ket{\Psi_0}), \ldots\} \nonumber \\
    &=\mathrm{Span}\{f^0(\ket{\Psi_0}), f^1(\ket{\Psi_0}), \ldots\, f^{m-1}(\ket{\Psi_0}) \nonumber \\
    &= \mathrm{K}_m 
    \label{eq:krylov1}
\end{align}
$m$ is called the grade of $\ket{\Psi_0}$ in regards to $f:=-iH$ \cite{PAR19, BAL22}.
Building upon this \cite{BAL22} proposed the spread complexity, which is a measure of how the state spreads over the Krylov basis of $\mathrm{K}_m$. This is achieved by orthonormalization through the Lanczos algorithm of the vectors 
$f^{j}(\ket{\Psi_0})$, which results in 
\begin{align}
\mathrm{K}_m=\mathrm{Span}\{\ket{k_0},..,\ket{k_{m-1}}\}.
\label{eq:krylov2}
\end{align}
The time evolved state is reconstructed through the basis representation $\ket{k_i}$ as
\begin{align}
    \ket{\Psi(t)} = \sum_{n=0}^{m-1}\bra{k_n}\ket{\Psi(t)}\ket{k_n}.
\end{align}
With $\alpha_n(t)=\bra{k_n}\ket{\Psi(t)}$, the spread complexity is defined as
\begin{align}
    \mathcal{K}_S(t) = \sum_{n=0}^{m-1} (n+1)\abs{\alpha_{n}(t)}^2.
    \label{eq:spread_complexity}
\end{align}
Spread complexity characterizes how an initial state evolves over a basis constructed from different powers of the Hamiltonian. This metric provides valuable insights and helps in the understanding of the time evolution operator,  but faces some challenges as an expressivity measure for quantum machine learning and quantum reservoir computing. First, the requirement for knowledge of the system's Hamiltonian poses a challenge because the quantum machine learning network is typically not described by a Hamiltonian but by a series of quantum circuits. Additionally, the different powers of the Hamiltonian must be computed classically, which becomes increasingly difficult as the number of qubits grows. To address these challenges, we demonstrated that time-evolved states can be used to construct a basis from which the complexity measure can be defined, exhibiting the same behavior in the researched systems \cite{CIN24a}. \\\\
\textbf{Krylov Expressivity}\\
In \cite{CIN24a} we demonstrated that instead of using the Krylov space 
\begin{align}
\mathrm{K}_m=\mathrm{Span}\Big(\ket{\Psi_0}, H\ket{\Psi_0},  \ldots, H^{m-1}\ket{\Psi_0}\Big),\nonumber 
\end{align}
a set of time-evolved states can be employed to construct the space
\begin{align}
\mathrm{G}_m:=\mathrm{Span}\Big(e^{-iHt_0}\ket{\Psi_0}, \ldots, e^{-iHt_{m-1}}\ket{\Psi_0}\Big).\nonumber 
\end{align}
The basis of $\mathrm{G}_m$ are evolved states and therefore quantum-mechanically measurable. The global phase does not change the space dimension and detailed proofs are given in \cite{CIN24a}. In \cite{CIN24a} we showed that the grade \( m \) of the Krylov space $\mathrm{K}_m$ is equal to the number of pairwise distinct eigenvalues of the Hamiltonian \( d \), i.e., \( m = d \). Building on this, Krylov expressivity is defined, which is upper bounded by the grade $m$. \\
The computation of \textbf{Krylov expressivity} requires the knowledge of the grade $m$. For this, time evolved states $\ket{g_i}=\exp (-iHt_i)\ket{\Psi_0}$ are computed, where some $T_K$ is picked and $t_i=(i+1)T_K/(N+1)$. For exactly the grade $m$, it will hold that $\ket{g_{m}}\in \mathrm{Span}(\ket{g_1},\ket{g_2},.., \ket{g_{m-1}})=\mathrm{G}_m$. We then simulate the time evolved state until some time $T$ and calculate the sampling times $\tau_i=(i+1)T/m$ with $i=0,..,m-1$. For small $T$ it is expected that the vectors $\ket{g_i}$ are close to each other, i.e. for all $i$, $\ket{g_i}\approx \ket{g_{i+1}}$ would hold. Analytically speaking $\ket{g_i}$ are independent, but numerically almost equal. For the computation of the Krylov expressivity, a measure of similarity between two time-evolved states $\ket{g_i}$ and $\ket{g_{i+1}}$ is required. Here, we proposed the computation of the fidelity for pure states $\mathrm{F}(\ket{g_i}, \ket{g_{i+1}})$ but other measures can also be utilized.
\begin{align}
    \lambda_{i} :=  \mathrm{F}\Big(\ket{g_i}\bra{g_i}, \ket{g_{i+1}} \bra{g_{i+1}} \Big)=\abs{\bra{g_i}\ket{g_{i+1}}}.
    \label{eq:fidelity}
\end{align}
The \textbf{Krylov expressivity} $\mathcal{E}_{K}$ for $\lambda$, where $\lambda \in [0,1]$ holds, is given by
\begin{align}
    m_{{eff}_i} &= \begin{cases}
        1 &\mathrm{if } \lambda_i < \lambda\\
        1-\frac{1}{1-\lambda}\cdot(\lambda_i-\lambda)  &\mathrm{if } \lambda_i \geq \lambda
    \end{cases} \nonumber \\
    \mathcal{E}_{K} &= 1 +\sum_{i=1}^{m-1}m_{{eff}_i}.
    \label{eq:eff_dim}
\end{align}
The first vector $\ket{g_1}$ adds a dimension of $1$. For $\lambda_i < \lambda$ we say that the two vectors are linearly independent and the Krylov expressivity is increased by one. In the region $\lambda_i \geq \lambda$, the Krylov expressivity $m_{{eff}_i}$ is interpolated to consider the difference between analytical independence and the independence that is needed for computation. We will explore the Krylov expressivity $\mathcal{E}_{K}$ in dependence of the time scale $T$ and which influence this has on the performance in quantum reservoir computing. For the following discussion, we set \( \lambda = \frac{1}{\sqrt{2}} \), inspired by a 3\,dB fall-off.

In Krylov complexity, the Krylov basis
\begin{align}
    \mathrm{K}_m = \mathrm{Span}\{\ket{\psi_0}, H\ket{\psi_0}, H^2\ket{\psi_0}, \ldots, H^{m-1}\ket{\psi_0}\}
\end{align}
is computed, and after normalization, the state at a given time $\ket{\psi(T)}$ is represented in the orthonormalized Krylov space. The complexity is then defined by how the state evolves within this Krylov space.

Krylov expressivity, on the other hand, defines a time-dependent Krylov space for each time $T$, given by
\begin{align}
    \mathrm{G}_m(T) = \mathrm{Span}\left\{\ket{\psi_0}, e^{iH\frac{T}{m}}\ket{\psi_0}, \ldots, e^{iHT}\ket{\psi_0} \right\}.
\end{align}
Krylov expressivity is then defined as the effective phase space dimension of $\mathrm{G}_m(T)$ and is therefor time-dependent. The same interpretation applies to Krylov operator complexity and Krylov observability, which will both be introduced in the following.

\subsection{Krylov Spaces for Operator Evolution}
Where spread complexity tries to quantify the spread of states, Operator complexity quantifies the spread of time-evolved operators \cite{PAR19} and is defined in the following paragraph.\\\\
\textbf{Operator complexity}\\
The evolution of operators is governed by  
\begin{align}
    \partial_t O(t) = i[H, O(t)].
\end{align}
The solution to this equation is
\begin{align}
    O(t) &= e^{iHt}Oe^{-iHt} = \sum_{k=0}^\infty \frac{(it)^k}{k!}\mathcal{L}^k(O),
\end{align}
where $\mathcal{L}(O)=HO-OH$ is the Liouvillian superoperator and linear in $O$, i.e.
\begin{align}
    \mathcal{L}(\lambda O_1+O_2)&=H(\lambda O_1+O_2)-(\lambda O_1+O_2)H \nonumber \\
    &= \lambda HO_1 - \lambda HO_1 + HO_2 - HO_2 \nonumber \\
    &= \lambda \mathcal{L}(O_1)+\mathcal{L}(O_2) 
\end{align}
Any time evolved operator $O(t)$ has therefore to be in the span of the powers of the Liouvillian $\mathcal{L}^n$
\begin{align}
    O(t) \in \mathrm{Span}\{\mathcal{L}^0(O),\mathcal{L}^1(O),\mathcal{L}^2(O),..\}.
\end{align}
By using the linear property of $\mathcal{L}$ and the Krylov space property, authors of 
\cite{PAR19} showed that there exists a $M\in \mathbb{N}$ such that 
\begin{align}
    O(t)\in \mathrm{Span}\{\mathcal{L}^0(O),\mathcal{L}^1(O),.., \mathcal{L}^{M-1}(O),..\}= \mathcal{L}_M
\end{align}
holds true. Similar to the Krylov state complexity the authors perform an orthonormalization procedure to construct the space
\begin{align}
    \mathcal{L}_M &:= \mathrm{Span}\{\mathcal{L}^0(O),\mathcal{L}^1(O),.., \mathcal{L}^{M-1}(O)\} \nonumber \\
    &=\mathrm{Span}\{\mathcal{W}_0,\mathcal{W}_1,.., \mathcal{W}_{M-1}\},
\end{align}
where $\{\mathcal{W}_i\}_{i=..,M-1}$ is the orthonormal basis constructed  through the Lanczos algorithm. In the next step the time evolved operator is expressed in the Krylov basis as 
\begin{align}
    O(t) = \sum_{n=0}^{M-1}i^n\beta_n(t)\mathcal{W}_n,
\end{align}
where $\beta_n(t) = (O(t), \mathcal{W}_n)$ is a scalar product defined on the operator space. Operator complexity $\mathcal{K}_O$ is then defined as
\begin{align}
    \mathcal{K}_O(t) = \sum_{n=0}^{M-1} (n+1)\abs{\beta_n(t)}^2.
    \label{eq:operator_complexity2}
\end{align}
 \\\\
 
\textbf{Krylov Observability}\\
One aspect that operator complexity lacks is a way to define an effective space dimension and a method to consider multiple operators. This is especially crucial for quantum machine learning, where a set of observables is measured, and quantum reservoir computing, where multiplexing is used to increase the readout dimension. In \cite{CIN25} we show that instead of using the Krylov space of the different powers of the Liouvillian $\mathcal{L}^k$
\begin{align}
    O(t) \in \mathcal{L}_M := \mathrm{Span}\{\mathcal{L}^0(O),\mathcal{L}^1(O),.., \mathcal{L}^{M-1}(O)\} \nonumber 
\end{align}
, we can use time-evolved observables to construct an equivalent space.
\begin{align}
    O(t) \in \mathcal{F}_M = \mathrm{Span}\{ \Tilde{O}(t_0),\Tilde{O}(t_1),.., \Tilde{O}(t_{M-1})\} = \mathcal{L}_M
\end{align} 

Building on this, we present an algorithm (\cref{alg:obs}) that enables the computation of a minimal space $\mathcal{F}$ such that any time-evolved operator $O_i(t) \in \mathcal{F}$ holds for all $O_i \in \{O_1, O_2, \dots, O_K\}$ and all $t$. This algorithm further returns the contribution of each operator $O_i$ as a space $\mathcal{F}_i$ with the following properties
\begin{align}\label{eq:Fprop}
    &\mathcal{F} = \bigcup_{k=1}^K \mathcal{F}_k = \bigcup_{k=1}^K \tilde{\mathcal{F}}_k~,~~  \mathrm{dim} \left( \bigcup_{j=1}^l \tilde{\mathcal{F}}_j \right) = \sum_{j=1}^l \mathrm{dim}(\tilde{\mathcal{F}}_j). \nonumber
\end{align}
The computation of \textbf{Krylov observability} considers \( K \) observables \( O_1, \dots, O_K \). For each observable \( O_k \), the spaces \( {\mathcal{F}}_k \) are computed with \( \mathrm{dim}({\mathcal{F}}_k) = M_k \). With \( R_k = \mathrm{min}(V, M_k) \) and \( \tau_{k} = T / M_k \), we can then define
$$\mathcal{G}_k = \{O_k(\tau_1), O_k(\tau_2), \dots, O_k(\tau_{R_k})\}, $$
where $V$ is the number of measurements. The observability of the \( k \)-th observable \( O_k \) is defined as 
\begin{align}
    \kappa_k(T) = 1 + \sum_{j=1}^{R_k - 1} (1 - \mathrm{F}(O_k(\tau_j), O_k(\tau_{j+1})))
\end{align}
, where  \( \mathrm{F} \) is the normalized fidelity between the two time-evolved operators given by
\begin{align}
    \mathrm{F}(A,B) = \abs{\mathrm{Tr}\Big(\frac{A^\dagger B}{\norm{A}\norm{B}}\Big)}.
    \label{eq:overlap_O}
\end{align}
The $\textbf{Krylov ~Observability}$ \( \mathcal{O}_K(T) \) of \( V \) multiplexed observables \( O_1, \dots, O_K \) is defined as  
\begin{align}
    \mathcal{O}_K(T) = \sum_{k=1}^K \kappa_k(T) \cite{CIN25}.
\end{align}

To compute the Krylov observability, we require the orthonormalization of the space $\mathcal{F}$. To construct a linearly independent basis, an algorithm such as the Gram–Schmidt procedure may be applied. In that case, each time-evolved operator at time $t_n$ must be orthonormalized with respect to all previous basis elements, requiring a total of $(n - 1)$ matrix multiplications for each $O(t_n)$. Given $V \cdot K$ matrices, the total number of matrix multiplications due to Gram–Schmidt is:
\begin{align}
     \sum_{n=1}^{VK} (n - 1) = \frac{VK(VK - 1)}{2}
\end{align}
Next, we require the computation of the time-evolved operators via $O(t_n) = e^{iHt} O e^{-iHt}$, which requires at most $VK$ matrix multiplications.
Lastly, for the computation of the overlaps $\mathrm{F}$ (\cref{eq:overlap_O}), another $VK$ matrix multiplications are required.
The resulting number of matrix multiplications required is then given by
\begin{align}
    N_{\mathrm{obs}} = \frac{VK(VK + 3)}{2}
    \label{eq:N_obs}
\end{align}

The ratio $r$ of the number of matrix multiplications for Krylov complexity to the construction of the state matrix is 
\begin{align}
    r = \frac{N_{\mathrm{obs}}}{N_{\mathrm{state}}} = \frac{V^2K^2 + 3VK}{2(7 + K)N_{\mathbf{u}}V} \in \mathcal{O}\left(\frac{VK}{N_{\mathbf{u}}}\right)
\label{eq:ratio}
\end{align}
In typical reservoir computing approaches, the number of inputs $N_{\mathbf{u}}$ is much larger than the readout dimension $N_R = VK$, i.e., $N_{\mathbf{u}} \gg VK$, thus resulting in $r \ll 1$.

\subsection{Dimension of Krylov Spaces}\label{subsec:dimension_KrylovSpace}

The construction of Krylov spaces is typically performed until a cutoff is observed in the algorithm, as discussed in \cref{app1:obs}. Due to numerical errors, the actual dimension considered may increase. Therefore, it is necessary to understand the grades \( m \) and \( M \) of the Krylov spaces 
\begin{align}
    \mathrm{K}_m &= \mathrm{Span}\{\ket{\psi_0}, H\ket{\psi_0}, \ldots, H^{m-1}\ket{\psi_0}\}, \nonumber \\
    \mathcal{L}_M &:= \mathrm{Span}\{\mathcal{L}^0(O), \mathcal{L}^1(O), \ldots, \mathcal{L}^{M-1}(O)\}.
\end{align}
It is also important to determine what these grades should be without relying on iterative algorithms like Lanczos, which are prone to numerical errors. In \cite{CIN24}, we showed that the grade \( m \) is upper bounded by the number of pairwise distinct eigenenergies and stated the following theorem.

\begin{theorem}
\label{theorem:E_d}
Let \( H \in \mathbb{C}^{N \times N} \) be a Hermitian Hamiltonian with \( d \) pairwise distinct eigenvalues \( \varepsilon_0, \varepsilon_1, \ldots, \varepsilon_{d-1} \), and let \( \{ \ket{\phi_j} \} \) be an orthonormal eigenbasis of \( H \), satisfying
\begin{align}
    H \ket{\phi_j} = \varepsilon_j \ket{\phi_j}.
\end{align}
Then the time-evolved state \( \ket{\Psi(t)} = e^{-iHt} \ket{\Psi_0} \) lies in a \( d \)-dimensional subspace \( \mathrm{E}_d \subseteq \mathbb{C}^N \), i.e.,
\begin{align}
    \ket{\Psi(t)} \in \mathrm{E}_d := \mathrm{Span} \left\{ \ket{\xi_0}, \ket{\xi_1}, \ldots, \ket{\xi_{d-1}} \right\},
\end{align}
where the vectors \( \ket{\xi_p} \) are defined by
\begin{align}
    \ket{\xi_p} := \frac{1}{\sqrt{|J_p|}} \sum_{j \in J_p} \alpha_j \ket{\phi_j}, ~~~ \mathrm{with~ } \alpha_j := \bra{\phi_j}\ket{\Psi_0},
\end{align}

and \( J_p := \{ j \mid \varepsilon_j = \varepsilon_p \} \) denotes the set of indices corresponding to the degenerate eigenspace of eigenvalue \( \varepsilon_p \). The normalization factor \( |J_p| \) is the cardinality of the set \( J_p \).
The tarting state \( \ket{\psi_0} \) can be represented in the basis \( \{\ket{\xi_p}\}_p \), with \( \gamma_p = \bra{\xi_p}\ket{\psi_0} \), as
\begin{align}
    \ket{\psi_0} = \sum_{p=0}^{d-1} \gamma_p \ket{\xi_p}.
\end{align}
Let \( n_1 \) denote the number of coefficients for which \( \gamma_p = 0 \). Then, the number of linearly independent vectors is reduced to \( d - n_1 \). It also holds that for the Krylov state space \( \mathrm{K}_m = \mathrm{Span}\{\ket{\psi_0}, H\ket{\psi_0}, \ldots, H^{m-1}\ket{\psi_0}\} \),
\begin{align}
    m = d - n_1
\end{align}
holds.
\end{theorem}

\begin{proof}
See \cref{app2:proof_Km_dim}.
\end{proof}
\medskip
The following theorem shows a similar identity for operator spaces $\mathcal{L}_M$. 
\begin{theorem}
\label{theorem:L_M_grade}
Let \( H \in \mathbb{C}^{N \times N} \) be a Hamiltonian with eigenbasis \( \{\ket{\phi_j}\} \) and corresponding eigenvalues \( \varepsilon_j \), and let \( O \) be an operator on the same Hilbert space. Define the Liouvillian Krylov space
\[
\mathcal{L}_M := \mathrm{Span}\{\mathcal{L}^0(O), \mathcal{L}^1(O), \ldots, \mathcal{L}^{M-1}(O)\},
\]
where \( \mathcal{L}(O) = [H, O] \) is the Liouvillian superoperator.

Define the transition frequencies \( \omega_{mn} := \varepsilon_m - \varepsilon_n \), and let \( \{\omega_P\}_{P=0}^{N_{\omega}-1} \) be the set of all pairwise distinct values taken by \( \omega_{mn} \). For each \( \omega_P \), define the index set
\[
J_P := \left\{ (m,n) \,\middle|\, \omega_{mn} = \omega_P \right\},
\]
and the corresponding matrix
\[
\sigma_P := \sum_{(m,n) \in J_P} \bra{\phi_m} O \ket{\phi_n} \ket{\phi_m} \bra{\phi_n}.
\]
Let \( N_1 \) is the number of vanishing contributions \( \sigma_P = 0 \), then the time-evolved operator is given by
\[
O(t) = \sum_{P \in S} e^{i\omega_P t} \sigma_P,
\]
with \( S = \{P \mid \sigma_P \neq 0\} = \{s_0, s_1, \ldots, s_{N_{\omega} - N_1 - 1}\} \). The operator lies in the span
\[
O(t) \in \mathrm{Span}\{\sigma_{s_0}, \sigma_{s_1}, \ldots, \sigma_{s_{N_{\omega} - N_1 - 1}}\}=\mathcal{P}_{N_\omega - N_1}.
\]
Further, the grade \( M \) of the Krylov space \( \mathcal{L}_M \) is given by
\[
M = N_{\omega} - N_1.
\]

\end{theorem}

\begin{proof}
See \cref{app3:proof_L_M_grade}.
\end{proof}

\section{Results}\label{sec:Results}

\begin{figure}[t]
\centering
    \hspace*{-0.5 cm}
    \centering
    \includegraphics[scale=1]{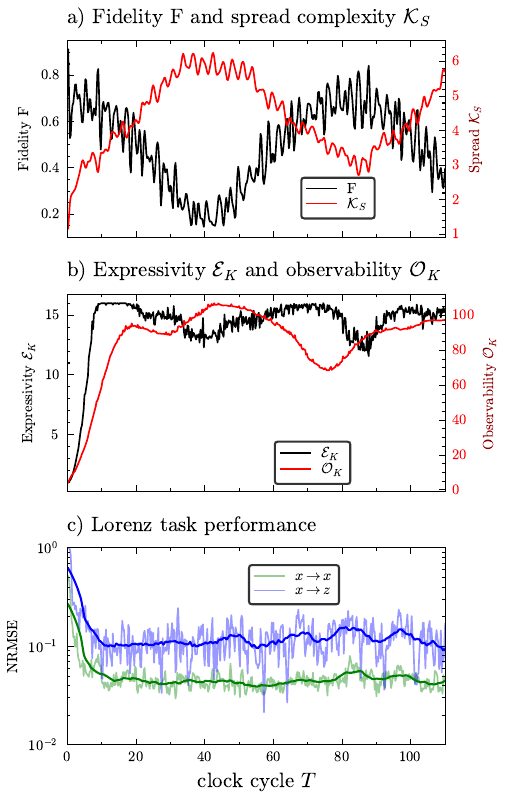}
    \caption[fidelity, spread complexity and Lorenz task performance]{Fidelity $\mathrm{F}$ (black) and spread complexity $\mathcal{K}_S$ (red) in (a), Krylov observability $\mathcal{O}_K$ (red) and Krylov expressivity $\mathcal{E}_K$ in (b) and the Lorenz task performance for the five-step ahead prediction ($\Delta t=0.1$) of the $x$-variable (green) and the cross prediction of the $z$- variable (blue) in (c) in dependence of the clock cycle $T$, when measuring one site of the $H_{I4}$ Hamiltonian. The bold line in (c) shows the first-order Savitzky--Golay filtered $\mathrm{NRMSE}$, used to better illustrate the saturation in task performance (computed using \texttt{scipy.signal.savgol\_filter}). }
    \label{fig:1_fid}
\end{figure}
The aim of this work is to utilize and compare expressivity measures to deepen the understanding of quantum reservoir computing. To begin, we simulate spread complexity and fidelity to investigate whether task performance can be explained by these measures. \cref{fig:1_fid}.a illustrates fidelity and spread complexity averaged over twenty random starting states. We see that spread complexity $\mathcal{K}_S$ increases, while fidelity $F$ decreases with larger clock cycles. Additionally, we observe an anti-proportional relationship between the two measures. When fidelity reaches a minimum at $T \approx 36$, spread complexity exhibits a maximum, and when fidelity increases, spread complexity decreases, and vice versa. 

The oscillations of $\mathrm{F}$ and $\mathcal{K}_S$ arise due to the periodic nature of the quantum system, where the state evolves as
\begin{align}
    \ket{\Psi(t)} := \sum_{n=0}^{N-1} e^{-i\varepsilon_n t} \ket{\phi_n} \bra{\phi_n}\ket{ \psi_0},
\end{align}
where $H\ket{\phi_n} = \varepsilon_n \ket{\phi_n}$ is the eigenvalue equation and $N = \dim(\mathcal{H})$ is the dimension of the Hilbert space. This is a superposition of $N$ periodic functions and is therefore periodic (or quasi-periodic) in nature.

The fidelity is given by $\mathrm{F} = \abs{\bra{\Psi(t)}\ket{\Psi_0}}$. Due to the quasi-periodic structure, there exist times $\tau$ for which $\ket{\Psi(\tau)} \approx \ket{\Psi_0}$. At such times, the fidelity $\mathrm{F} = \abs{\bra{\Psi(\tau)}\ket{\Psi_0}}$ will be close to one, which is observed at $\tau \approx 80$.

On the other hand, Krylov spread complexity $\mathcal{K}_S$ measures the spread over the Krylov basis. The first basis state represents the initial state, $\ket{k_0} = \ket{\Psi_0}$. In the construction of Krylov spread complexity (\cref{eq:spread_complexity}), the amplitude with respect to $\ket{k_0}$ contributes with the weakest weight of $1$ and is given by $\abs{\alpha_0(t)} = \abs{\bra{\Psi(t)}\ket{\Psi_0}} = \mathrm{F}$. Since the normalization condition $\sum_n \abs{\alpha_n(t)}^2 = 1$ holds, a large value of $\abs{\alpha_0(t)}$ implies that the contributions from the other $\abs{\alpha_n(t)}$ must be smaller. This, in turn, implies that the amplitudes of the higher-weight contributions are reduced, thereby lowering $\mathcal{K}_S$. This explains the dip in Krylov spread complexity, which varies inversely with the fidelity.

In \cref{fig:1_fid}.c, the $\mathrm{NRMSE}$ values of the Lorenz tasks for the five-step prediction of the $x$ variable (green) and the cross-prediction of the $z$ variable (blue) are presented. We observe a reduction in error, followed by a saturation in task performance. The initial decrease in $\mathrm{NRMSE}$ can be explained using fidelity, which captures how far the state evolves from the initial state. This however only tells us that the system has to evolve, to perform any operations as at $T=0$ the reservoir evolution is simply the identity. Meanwhile, spread complexity quantifies how the time-evolved state spreads over time in the Krylov basis. Higher spread complexity therefore implies a broader spread within the Krylov basis at the corresponding clock cycle. However, neither fidelity nor spread complexity alone can explain the saturation of $\mathrm{NRMSE}$. The oscillatory behavior observed in fidelity and spread complexity (\cref{fig:1_fid}.a is not reflected in task performance (\cref{fig:1_fid}.c.
\cref{fig:1_fid}.b shows Krylov expressivity $\mathcal{E}_K$ and Krylov observability $\mathcal{O}_K$. Both measures show an increase followed by saturation, around which they oscillate, similar to the task performance in \cref{fig:1_fid}.c. This shows that Krylov expressivity and Krylov observability capture the saturation in task performance, while only limited insights can be gained from fidelity and spread complexity.
\begin{figure*}[t]
\hspace*{-1cm}
    \includegraphics[scale=1]{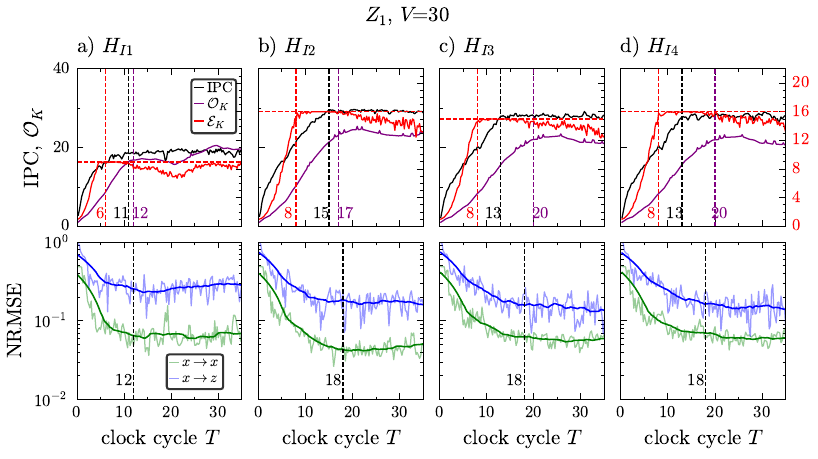}
    \caption[Expressivity, Observability and Lorenz task performance]{The first row shows the information processing capacity $\mathrm{IPC}$ (black), the Krylov expressivity $\mathcal{E}_K$ (red) and the Krylov observability $\mathcal{O}_K$ (purple) and the second row shows the Lorenz task performance for the five-step ahead prediction ($\Delta t=0.1$) of the $x$-variable (green) and the cross prediction of the $z$- variable (blue) in dependence of the clock cycle $T$ for the different quantum reservoirs \textbf{a)} $H_{I1}$, \textbf{b)} $H_{I2}$,\textbf{c)} $H_{I3}$ and \textbf{d)} $H_{I4}$. The state matrix is constructed by measuring the $Z_1$ observable $V=30$ times. The vertical dotted lines in the first row represent when IPC (black), $\mathcal{O}_K$ (purple) and $\mathcal{E}_K$ (red) saturate and the time is given next to the line. The vertical black line in the second row, indicates when the Lorenz task performance saturates. The bold line in (c) shows the first-order Savitzky--Golay filtered $\mathrm{NRMSE}$, used to better illustrate the saturation in task performance (computed using \texttt{scipy.signal.savgol\_filter}).}
    \label{fig:2_Z1}
\end{figure*}

\subsection{Specific Hamiltonians}
\cref{fig:2_Z1} shows the results when only the first observable $Z_1$ is measured, while in \cref{fig:2_Z1234}, measurements of the observables $Z_1$, $Z_2$, $Z_3$, and $Z_4$ are considered. In all cases, the observables are multiplexed $V=30$ times. A summary of the results is provided in \cref{appendix:saturation}. The Krylov expressivity $\mathcal{E}_K$ of each Hamiltonian $H_{I1}$, $H_{I2}$, $H_{I3}$, and $H_{I4}$ is upper-bounded by $9$, $16$, $15$, and $16$, respectively. This upper bound corresponds to the grade of the Krylov space or the number of pairwise distinct eigenvalues. \cref{fig:2_Z1}.a shows the result for $H_{I1}$. We observe that $\mathrm{NRMSE}$ of the Lorenz task, $\mathrm{IPC}$, and $\mathcal{O}_K$ all saturate around $T=12$, while $\mathcal{E}_K$ reaches saturation at $T=6$.
Increasing the complexity of the Hamiltonian slightly by modifying one of the inter-spin couplings $J_{i.j}$ results in $H_{I2}$. This Hamiltonian reaches a higher $\mathrm{IPC}$, Krylov expressivity $\mathcal{E}_K$ and Krylov observability $\mathcal{O}_K$. Furthermore the Lorenz task errors are slightly decreased. 
When using pairwise distinct inter-spin couplings $J_{i,j}$ with equal spacing, the resulting Hamiltonian is $H_{I3}$ (\cref{fig:2_Z1}.c. Here, $\mathcal{E}_K$, $\mathrm{IPC}$, and $\mathcal{O}_K$ each saturate at $T_{\mathrm{sat}} \approx 8$, $13$, and $20$, respectively. The Lorenz task performance, however, reaches saturation at $T_{\mathrm{sat}} \approx 18$, which best aligns with the behavior of  Krylov observability $\mathcal{O}_K$.\\
For $H_{I1}$ and $H_{I2}$, $\mathrm{IPC}$ and $\mathcal{O}_K$ show similar behavior, but a discrepancy appears in $H_{I3}$. This occurs because $\mathrm{IPC}$ is upper-bounded by the readout dimension of $N_R=30$, to which the $\mathrm{IPC}$ then saturates $\mathrm{IPC}\leq N_R$. In this case,  Krylov observability $\mathcal{O}_K$ best explains the behavior of task performance.
The worse task performance of $H_{I1}$ compared to the other Hamiltonians is explained by the fact that all metrics $\mathrm{IPC}$, $\mathcal{E}_K$, and $\mathcal{O}_K$ are all smaller. \\
When comparing $H_{I3}$ and $H_{I2}$ we note that $\mathcal{E}_K(H_{I3})<\mathcal{E}_K(H_{I2})$. At the same time, we observe that the task performance for the cross-prediction task is better for $H_{I3}$, while it is worse for the prediction of the $x$ variable. To further investigate this discrepancy, we introduce $H_{I4}$. This Hamiltonian is nearly identical to $H_{I3}$, differing only in one inter-spin coupling, which results in a Krylov expressivity equal to that of $H_{I2}$. Despite this adjustment, the $\mathrm{IPC}$, $\mathcal{O}_K$, $\mathcal{E}_K$, and Lorenz task errors for $H_{I4}$ remain nearly identical to those of $H_{I3}$. 
\begin{figure*}[t]
\hspace*{-1cm}
    \includegraphics[scale=1]{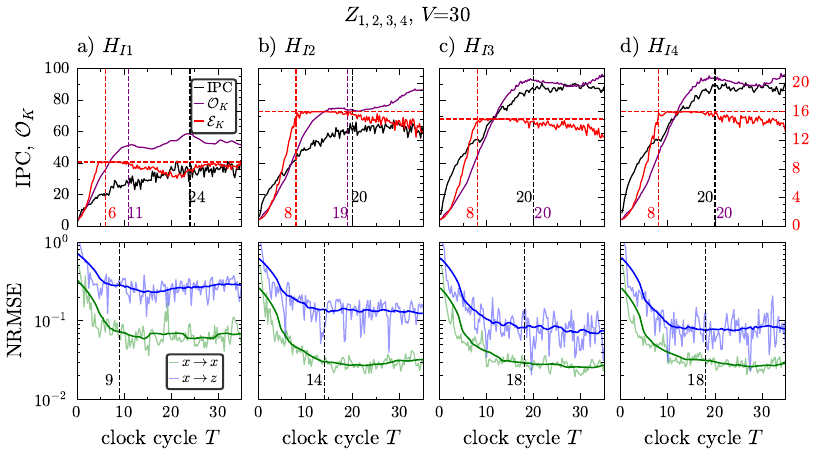}
    \caption[Expressivity, Observability and Lorenz task performance]{Same as in \cref{fig:2_Z1}, except that the state matrix is constructed by measuring the $Z_1$, $Z_2$, $Z_3$, and $Z_4$ observables $V = 30$ times.}
    \label{fig:2_Z1234}
\end{figure*}

To discuss this further, we consider a typical quantum reservoir where each observable is measured, as shown in \cref{fig:2_Z1234}. This results in a readout dimension of $N_R=120$. In this setup, $H_{I1}$ again exhibits largest errors. While the Lorenz task performance remains unchanged, the information processing capacity increases from $\mathrm{IPC}=20$ (\cref{fig:2_Z1}) to $\mathrm{IPC}=40$ in \cref{fig:2_Z1234}.a.
For $H_{I2}$, $H_{I3}$, and $H_{I4}$ (\cref{fig:2_Z1234}.b-d, we observe that $\mathrm{IPC}$ and $\mathcal{O}_K$ both perform similarly well in explaining the saturation of the $\mathrm{NRMSE}$ of the Lorenz task. \\
The improved performance of $\mathrm{IPC}$ in \cref{fig:2_Z1234} compared to \cref{fig:2_Z1} is due to the system being sampled with $N_R=120$ readout nodes. In this case the $\mathrm{IPC}$ is upper-bounded by $\mathrm{IPC}\leq N_R=120$, which is sufficiently large to not observe the upper bound by the readout dimension.
\begin{table}[H]
\centering
\begin{tabular}{lcccc}
\toprule
 & 
 \begin{tabular}[c]{@{}c@{}}$Z_1$\\ $x \rightarrow x$\\ $\mathrm{NRMSE}_{\mathrm{sat}}$\end{tabular} &
 \begin{tabular}[c]{@{}c@{}}All sites\\ $x \rightarrow x$\\ $\mathrm{NRMSE}_{\mathrm{sat}}$\end{tabular} &
 \begin{tabular}[c]{@{}c@{}}$Z_1$\\ $x \rightarrow z$\\ $\mathrm{NRMSE}_{\mathrm{sat}}$\end{tabular} &
 \begin{tabular}[c]{@{}c@{}}All sites\\ $x \rightarrow z$\\ $\mathrm{NRMSE}_{\mathrm{sat}}$\end{tabular} \\
\midrule
$H_1$ & 0.08 & 0.08 & 0.30 & 0.30 \\
$H_2$ & 0.04 & 0.03 & 0.20 & 0.15 \\
$H_3$ & 0.06 & 0.03 & 0.16 & 0.08 \\
$H_4$ & 0.06 & 0.03 & 0.16 & 0.08 \\
\bottomrule
\end{tabular}
\caption{$\mathrm{NRMSE}$ for the Lorenz task of the four Hamiltonians $H_{I1},~H_{I2},~H_{I3}$ and $H_{I4}$, when measuring one site $Z_1$ (row one and three) and when measuring all sites (row two and four).}
\label{tab:error123}
\end{table}

\cref{tab:error123} shows the $\mathrm{NRMSE}$ of the Lorenz task at saturation, when only $Z_1$ is measured and when all four sites are measured. For $H_{I1}$, we observe that the number of sites measured does not change task behavior. For $H_{I2}$, we observe a slightly smaller $\mathrm{NRMSE}$ when all sites are measured. $H_{I3}$ and $H_{I4}$ exhibit identical task performance, with the error of both tasks halved when more observables are measured.
To discuss the differing behavior of expressivity and observability between $H_{I2}$ and $H_{I3}$, we take a closer look at task performance when all sites are measured (\cref{fig:2_Z1234}). We observe that the cross-prediction task for $H_{I3}$ yields an error of $\mathrm{NRMSE} = 0.08$, while for $H_{I2}$ the error is higher, at $\mathrm{NRMSE} = 0.15$. The Krylov expressivity for $H_{I2}$, with $\mathcal{E}_K(H_{I2}) = 16$, is slightly higher than that of $H_{I3}$, which has $\mathcal{E}_K(H_{I3}) = 15$. 

At first glance, this may seem counterintuitive: $H_{I2}$ maps input data into a larger Krylov space, suggesting that more information should be accessible. However, although the input is projected into a higher-dimensional space, information can only be extracted through the set of measured observables. In this case, the number of independent measurements for $H_{I2}$ is actually smaller than for $H_{I3}$, which is mirrored by the smaller Krylov observability $\mathcal{O}_K(H_{I2}) = 60$ versus $\mathcal{O}_K(H_{I3}) = 85$.

To further probe this behavior, the $\mathrm{IPC}$ and $\mathcal{O}_K$ are computed for various clock cycles, virtual nodes, and numbers of readout observables for $H_{I2}$ (\cref{fig3:obsH2}) and $H_{I3}$ (\cref{fig3:obsH3}). In the first and second rows, only the first site and the first two sites are measured for the construction of the state matrix, respectively, while the third row shows the results when all sites are measured.
For $H_{I2}$ (\cref{fig3:obsH2}), we observe $\mathrm{IPC}<85$ and Krylov observability $\mathcal{O}_K<85$. In contrast, $H_{I3}$ shows a maximum of $105$ (\cref{fig3:obsH3}), indicating that the larger $\mathrm{IPC}$ is due to the increased Krylov observability $\mathcal{O}_K$.
Furthermore, an increase in the number of observables measured shifts the best performance from larger clock cycles and a larger number of measurements to smaller clock cycles and fewer measurements. Each plot includes the Pearson correlation coefficient $P_C$ (\cref{eq:Pearson}) between $\mathrm{IPC}$ and $\mathcal{O}_K$. For $H_{I2}$, the correlation coefficients are $P_C = 0.94$ when one or two observables are measured, and $P_C = 0.96$ when all observables are considered. This indicates a strong and consistent correlation across all configurations. The corresponding correlation coefficients for $H_{I3}$ are $P_C = 0.91$ for one observable and $P_C = 0.95$ for both two and four observables.
\begin{figure}[t]
\centering
    \hspace*{-0.5 cm}
    \centering
    \includegraphics[scale=01]{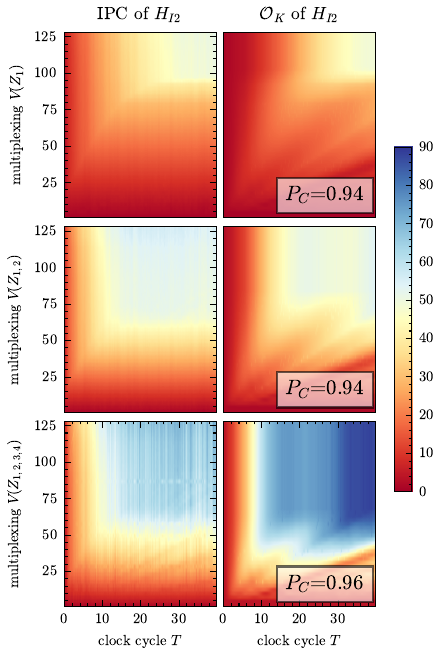}
    \caption[Observability]{Observability $\mathcal{O}_K$ (right column) and $\mathrm{IPC}$ (left column) are color-coded according the the color bar, depending on multiplexing $V$ and clock cycle $T$ for $H_{I2}$. The Pearson correlation factor $P_C$ between the two images in each row is calculated, indicating an almost identical behavior between $\mathcal{O}_K$ and $\mathrm{IPC}$.}
    \label{fig3:obsH2}
\end{figure}

\begin{figure}[t]
\centering
    \hspace*{-0.5 cm}
    \centering
    \includegraphics[scale=01]{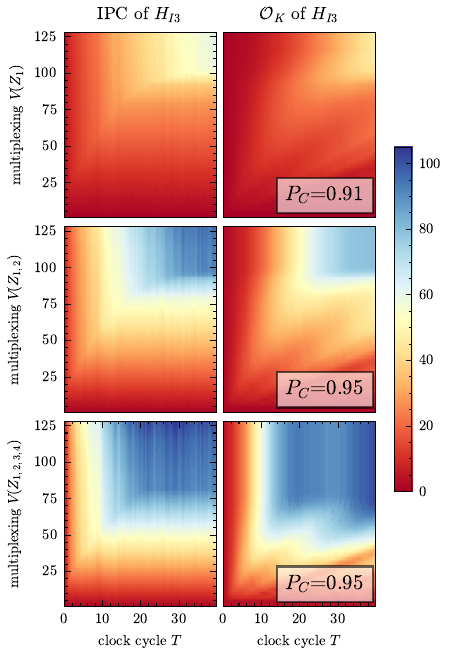}
    \caption[Observability]{Same as in \cref{fig3:obsH3} for $H_{I3}$.}
    \label{fig3:obsH3}
\end{figure}

\begin{table}[H]
\centering
\begin{tabular}{llccccc}
\toprule
\textbf{Hamiltonian} & \textbf{Observable} & $~d~$ & $~d^2~$ & $~N_{\omega}~$ & $~N_1~$ & $~M~$ \\
\midrule
$H_{I2}$ & $Z_1$ & 16 & 256 & 237 & 176 & 61 \\
$H_{I2}$ & $Z_2$ & 16 & 256 & 237 & 176 & 61 \\
$H_{I2}$ & $Z_3$ & 16 & 256 & 237 & 158 & 79 \\
$H_{I2}$ & $Z_4$ & 16 & 256 & 237 & 158 & 79 \\
\midrule
$H_{I3}$ & $Z_1$ & 15 & 225 & 211 & 112 & 99 \\
$H_{I3}$ & $Z_2$ & 15 & 225 & 211 & 112 & 99 \\
$H_{I3}$ & $Z_3$ & 15 & 225 & 211 & 112 & 99 \\
$H_{I3}$ & $Z_4$ & 15 & 225 & 211 & 112 & 99 \\
\bottomrule
\end{tabular}
\caption{Spectral and dynamical statistics for each (Hamiltonian, Observable) pair. $d$ is the number of pairwise distinct eigenvalues, $d^2$ is the number of eigenvalue pairs, $N_{\omega}$ the number of distinct transition frequencies, $N_1$ the number of zero $\sigma_P$ contributions, and $M$ the resulting Krylov dimension.}
\label{tab:spectral_summary}
\end{table}
While Krylov observability captures the discrepancy between \( H_{I2} \) and \( H_{I3} \), a deeper explanation is provided by the analysis in \cref{subsec:dimension_KrylovSpace}. There, the Krylov grades of \( \mathcal{L}_M \) and \( \mathrm{K}_m \) are defined in terms of the number of distinct eigenvalues \( d \), transition frequencies \( N_{\omega} \), and vanishing contributions \( N_1 \), where \( \sigma_P = 0 \). The rank of the Liouvillian Krylov space is given by \( M = N_{\omega} - N_1 \).
We compute these quantities for all observables \( Z_1 \) to \( Z_4 \) and summarize the results in \cref{tab:spectral_summary}. Although \( H_{I2} \) has more distinct eigenvalues (\( d = 16 \)) and frequencies (\( N_{\omega} = 237 \)), it also exhibits significantly more vanishing contributions: \( N_1 = 176 \) for \( Z_1 \), \( Z_2 \) and \( N_1 = 158 \) for \( Z_3 \), \( Z_4 \). In contrast, \( H_{I3} \) shows fewer zero terms with \( N_1 = 112 \) and \( N_{\omega} = 211 \) across all observables.

The resulting Krylov grades are \( M(H_{I2}) = 61 \) for \( Z_1 \) and \( Z_2 \), and \( M = 79 \) for \( Z_3 \) and \( Z_4 \), while \( H_{I3} \) consistently yields \( M = 99 \) across all observables. This supports the higher expressivity of \( H_{I3} \), as reflected by its larger \( \mathrm{IPC} \) values, and highlights how differences in Krylov space dimensions influence performance. Results for the remaining Hamiltonians are provided in \cref{app4:KrylovSpaceDim}.


In a final experiment, we compute the $\mathrm{IPC}$ and $\mathcal{O}_K$ for an Ising Hamiltonian with five sites and pairwise distinct couplings. The results are presented in \cref{fig4:Ns5}, where one site (first row), three sites (second row), and all five sites (third row) are measured in the $z$-direction. Simulations are performed up to $V = 220$ virtual nodes and clock cycles up to $T = 200$. Computing the $\mathrm{IPC}$ for each parameter configuration requires hundreds of hours for the five-site system and tens of hours for the four-site system, which limits our ability to simulate even larger systems. Once again, we observe almost identical behavior between $\mathrm{IPC}$ and $\mathcal{O}_K$, as indicated by a correlation factor of $P_C = 0.97$.

\begin{figure}[t]
\centering
    \hspace*{-0.5 cm}
    \centering
    \includegraphics[scale=01]{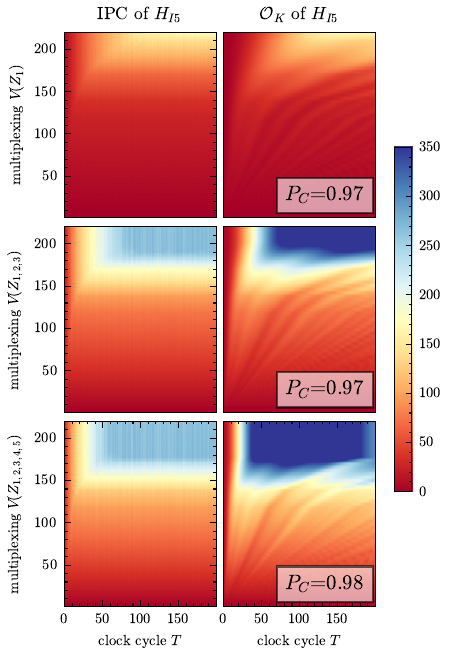}
    \caption[Observability]{Same as in \cref{fig3:obsH3} for for a reservoir with five sites $N_S=5$.}
    \label{fig4:Ns5}
\end{figure}

\subsection{Hamiltonian with Random Coupling}
Until now, specific Hamiltonians have been analyzed with respect to task performance, typically restricted to very small systems in order to describe behavior across the full parameter space. However, most research simulates larger spin systems with random inter-spin coupling. The inter-spin couplings of such quantum reservoirs can be sampled from distributions to compute statistics and reveal general trends in task performance. 

In this regard, we now analyze the behavior of the transverse field Ising model \cref{eq:Ising_Ham}, where the inter-spin couplings are sampled uniformly in the interval $[0.25, 0.75]$, i.e., $J_{ij} \in \mathcal{U}([0.25, 0.75])$, with $h=0.5$. These parameters are commonly used for QRC purposes \cite{FUJ17}. The state matrix is constructed by measuring all spin sites of a system with $N_S=6$ sites, repeated $V$ times.

\begin{figure*}[t]
\hspace*{-1cm}
    \includegraphics[scale=1]{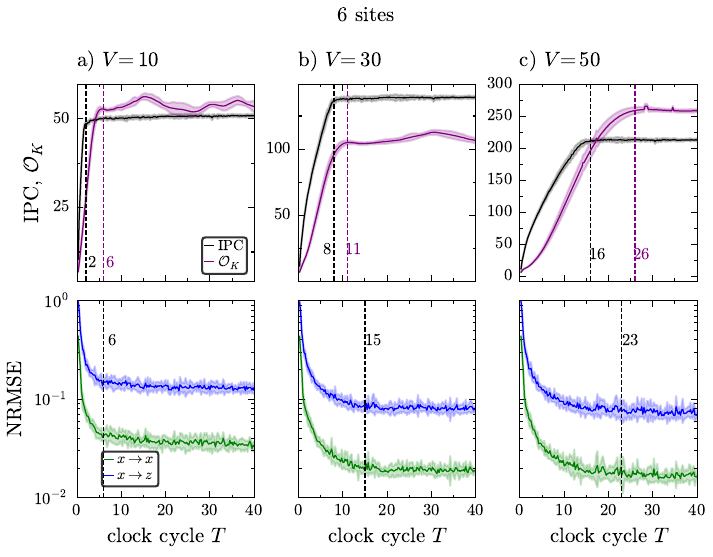}
    \caption[Expressivity, Observability and Lorenz Task Performance]{The first row shows the information processing capacity (IPC) in black and the Krylov observability $\mathcal{O}_K$ in purple. The second row shows the Lorenz task performance: the five-step ahead prediction ($\Delta t=0.1$) of the $x$-variable (green) and the cross prediction of the $z$-variable (blue), both plotted against the clock cycle $T$, for different numbers of measurements: \textbf{a)} $V=10$, \textbf{b)} $V=30$, and \textbf{c)} $V=50$. The state matrix is constructed by measuring all sites in the Pauli-$z$ direction. The inter-spin couplings are sampled uniformly with $J_{ij} \in \mathcal{U}([0.25, 0.75])$. The bold lines in this plot represent the average over 10 different Hamiltonians, and the shaded regions indicate the standard deviation of each curve.}
    \label{fig:review_1}
\end{figure*}

Figure~\ref{fig:review_1} presents the results of simulating the transverse field Ising model with random inter-spin couplings. Three configurations for the number of measurements are shown: \( V = 10 \) (a), \( V = 30 \) (b), and \( V = 50 \) (c). For each case, 10 Hamiltonians were randomly sampled, and the average (bold line) and standard deviation (shaded fill) were computed. The top row in each subfigure displays the information processing capacity (IPC) in black and the Krylov observability \( \mathcal{O}_K \) in purple.

The bottom row shows task performance for the Lorenz forecasting task. The green curve represents the five-step ahead prediction of the \( x \)-variable with a prediction step size of \( \Delta t = 0.1 \), while the blue curve shows the cross prediction of the \( z \)-variable from the measured \( x \)-dynamics.

With $V=10$ measurements (\cref{fig:review_1}.a), both Krylov observability $\mathcal{O}_K$ and the information processing capacity IPC exhibit a rapid increase and reach visible saturation at $T=6$ and $T=2$, respectively. Task performance also decreases rapidly and reaches a saturation region around $T=6$, similarly to Krylov observability.

Increasing the number of measurements to $V=30$ yields the results shown in the second row (\cref{fig:review_1}.b). Here, IPC reaches a maximum of approximately $170$ at $T=8$, likely due to the increased number of readout nodes. $\mathcal{O}_K$ saturates around $T=11$, which better corresponds to the saturation of Lorenz task performance observed around $T=15$.

Lastly, the system is measured $V=50$ times, with results shown in the third row (\cref{fig:review_1}.c). In this case, IPC and $\mathcal{O}_K$ reach saturation at $T=16$ and $T=26$, respectively, while visible saturation in Lorenz task performance occurs at $T=23$. Once again, Krylov observability better captures the trend in task performance compared to IPC, especially when the system is undersampled, thereby showcasing the generality of these results.

\subsection{Computational Cost}
We conclude this section with a brief discussion on the computational cost associated with constructing the state matrix $\mathbf{S}$ for the computation of IPC and Krylov observability, as discussed in \cref{sec:QRC} and \cref{sec:Krylov}.

In our case, with $N_{\mathrm{Tr}} = 25{,}000$, $N_{\mathrm{Te}} = 5{,}000$, $V = 30$, and four measured observables ($K = 4$), the total number of matrix multiplications required for constructing the state matrix is given by \cref{eq:N_state} to
\begin{align}
N_{\mathrm{state}} &= (7 + K) N_{\mathbf{u}} V \nonumber \\
&= (7 + 4) \cdot 30{,}000 \cdot 30 = 9.9 \cdot 10^{6}.
\end{align}
The number of matrix multiplications required for computing Krylov observability by \cref{eq:N_obs} is:
\begin{align}
N_{\mathrm{obs}} &= \frac{VK(VK + 3)}{2} \nonumber \\
&= \frac{30 \cdot 4 (30 \cdot 4 + 3)}{2} = 7380 \quad  
\end{align}
The computation of Krylov observability therefore requires only
\[
r = \frac{N_{\mathrm{obs}}}{N_{\mathrm{state}}} = 0.75\%
\]
of the matrix multiplications required for constructing the state matrix, demonstrating a three-orders-of-magnitude speed-up.  It is important to note that this estimate does not include the actual training across thousands of tasks, which involves inversion of the state matrix and optimization of readout weights. These steps are typically more computationally intensive than the construction of the state matrix itself, further amplifying the efficiency gains offered by the Krylov operator complexity compared to IPC.


\section{Discussion and Conclusion}
This work investigates Krylov-based information measures to explain and predict task performance in quantum reservoir computing, using the Lorenz63 system for chaotic time-series prediction. Initially, fidelity and spread complexity are computed, but they fail to explain task performance saturation \cref{fig:1_fid}. \\
We then extended our research to Krylov expressivity $\mathcal{E}_K$ \cite{CIN24a} and Krylov observability $\mathcal{O}_K$ \cite{CIN25} and show that both measures show saturation for larger clock cycles. 
The measures are compared to the information processing capacity $\mathrm{IPC}$, which captures how well the system can retain and map data non-linearly. We test four quantum reservoirs in an undersampled regime, where Krylov observability $\mathcal{O}_K$ outperforms $\mathrm{IPC}$ in explaining the trend in task performance. In this case, $\mathrm{IPC}$ reaches a maximum, while the error for the Lorenz tasks continues to decrease with larger clock cycles $T$.\\
The quantum reservoirs are then simulated when all sites are measured, as shown in \cref{fig:2_Z1234}. The first observation from this experiment is that $\mathrm{IPC}$ and $\mathcal{O}_K$ exhibit almost identical behavior. Another significant result is the superior Lorenz task performance and larger $\mathrm{IPC}$ of $H_{I3}$ compared to $H_{I2}$, despite $H_{I3}$ having a smaller Krylov expressivity. While Krylov expressivity provides insights into the space onto which the data is mapped, the amount of information extracted from that space is smaller for $H_{I2}$, as evidenced by the smaller Krylov observability $\mathcal{O}_K$ of $H_{I2}$ compared to $H_{I3}$. To better understand this, we discuss the grade of Krylov state spaces in accordance with \cite{CIN24a} in \cref{theorem:E_d}, and introduce a similar identity for Krylov operator spaces in \cref{theorem:L_M_grade}. This is then used to show that the smaller Krylov operator space of $H_{I2}$ is due to an increased number of zero contributions compared to $H_{I3}$, which explains the overall larger Krylov observability for $H_{I3}$.

To confirm that our measures are not dependent on system size, we computed the information processing capacity and Krylov observability for a five-qubit system over various clock-cycles and number of measurements, as shown in \cref{fig4:Ns5}, achieving correlation factors between IPC and Krylov observability of $P_C \geq 0.97$. We then simulate a six-site Ising model with random inter-spin couplings, as is commonly done in quantum reservoir computing. This allows for statistical analysis of these measures, where we show that the results match the behavior observed in the four- and five-site systems, in which Krylov observability better explains the trend in task behavior when the system is undersampled (\cref{fig:review_1}).
Since Krylov expressivity considers how input states are mapped onto the Krylov space, information about the input encoding can be gained. In quantum reservoir computing, this might not be of much importance, because the state evolves over time. In quantum machine learning, however, where input encoding is one of the main parts to be optimized, Krylov expressivity can be used to effectively understand and compare various encoding strategies. One approach would be to encode the input data $x_i$ through a unitary circuit $U_E(x_i)$ into the quantum reservoir or quantum machine learning network. The initial states in the computation of Krylov expressivity can then be sampled from the set of encoded states, i.e., $\ket{x} \in \{ U_E(x_i)\ket{s} \mid x_i \in \mathcal{X} \}$, where $\mathcal{X}$ is the set of inputs.
The utility of Krylov observability might be used to further gain understanding in quantum machine learning, such as barren plateaus or as a quantum-mechanical information measure that can be utilized in the understanding of quantum dynamics, further advancing knowledge in chaos, decoherence, or thermalization.

\appendix
\section{Algorithm for the construction of the spaces}\label{app1:obs}
\begin{algorithm}[H]]\label{alg:cap}
\caption{Construction of Observability Spaces}
\begin{algorithmic}[1]
\State \( \mathcal{I}_O = \{1, 2, \dots, K\} \)
\State \( \mathcal{T} = \{t_1, t_2, \dots, t_R\} \)
\State \( \mathcal{F}^{(B)} = \emptyset \)
\State \( \tilde{\mathcal{F}}_1^{(B)}, \tilde{\mathcal{F}}_2^{(B)}, \dots, \tilde{\mathcal{F}}_K^{(B)} = \emptyset, \emptyset, \dots, \emptyset \)
\While {$t_j \in \mathcal{T}$}
    \For{$k \in \mathcal{I}_O $}
        \If {$O_k(t_j) \notin \mathcal{F}$} 
        \State \( \mathcal{F}^{(B)} = \mathcal{F}^{(B)} \cup O_k(t_j) \)
        \State \( \tilde{\mathcal{F}}_k^{(B)} = \tilde{\mathcal{F}}_k^{(B)} \cup O_k(t_j) \)
        \EndIf
    \EndFor
\EndWhile
\State \( \mathcal{F}^{(B)} \leftarrow \mathrm{ONB}(\mathcal{F}^{(B)}) \)
\end{algorithmic}
\end{algorithm}
Consider the basis \( \mathcal{F}^B \) and \( \mathcal{F}^B_i \) of the spaces \( \mathcal{F} \) and \( \mathcal{F}_i \), i.e., \( \mathcal{F} = \mathrm{Span}(\mathcal{F}^B) \) and \( \mathcal{F}_i = \mathrm{Span}(\mathcal{F}^B_i) \). Then \cref{alg:cap} constructs spaces such that the properties given in \cref{eq:Fprop} are defined.

\section{Saturation}
\label{appendix:saturation}
\cref{table3:summary} gives information when each measure exhibits a saturation point.
\begin{table*}
\begin{tabular}{lcccccccc}
\hline
 &
  \begin{tabular}[c]{@{}c@{}}$\mathcal{E}_K$\\ $T_{\mathrm{sat}}$\end{tabular} &
  \begin{tabular}[c]{@{}c@{}}$\mathcal{O}_K$\\ $T_{\mathrm{sat}}$\end{tabular} &
  \begin{tabular}[c]{@{}c@{}}$\mathrm{IPC}$\\ $T_{\mathrm{sat}}$\end{tabular} &
  \begin{tabular}[c]{@{}c@{}}$\mathrm{IPC}$\\ saturation point\end{tabular} &
  \begin{tabular}[c]{@{}c@{}}$x\rightarrow x$\\ $T_{\mathrm{sat}}$\end{tabular} &
  \begin{tabular}[c]{@{}c@{}}$x\rightarrow z$ \\ $T_{\mathrm{sat}}$\end{tabular} &
  \begin{tabular}[c]{@{}c@{}}$x\rightarrow x$,\\ $\mathrm{NRMSE}_{\mathrm{sat}}$\end{tabular} &
  \begin{tabular}[c]{@{}c@{}}$x\rightarrow z$\\ $\mathrm{NRMSE}_{\mathrm{sat}}$\end{tabular} \\ \hline\hline
$H_1$, $Z_1$             & 6 & 12 & 12 & 20 & 12                         & 12                         & 0.08 & 0.3  \\ 
$H_1$, all $Z_i$ & 6 & 10 & 34 & 40 & 12                         & 12                         & 0.08 & 0.3  \\ \hline
$H_2$, $Z_1$             & 8 & 17 & 15 & 30 & 18                         & 18                         & 0.04 & 0.2  \\ 
$H_2$, all $Z_i$ & 8 & 20 & 22 & 60 &  18 &  18 & 0.03 & 0.15 \\ \hline
$H_3$, $Z_1$             & 8 & 20 & 13 & 30 & 18                         & 18                         & 0.06 & 0.16 \\ 
$H_3$, all $Z_i$ & 8 & 20 & 20 & 87 & 20                         & 20                         & 0.03 & 0.08 \\ \hline
$H_4$, $Z_1$             & 8 & 20 & 13 & 30 & 18                         & 18                         & 0.06 & 0.16 \\ 
$H_4$, all $Z_i$ & 8 & 20 & 20 & 87 & 20                         & 20                         & 0.03 & 0.08 \\ \hline
\end{tabular}
\caption{Information about the saturation behavior, saturation time and saturation points of the different Hamiltonians.}
\label{table3:summary}
\end{table*}

\section{Proofs for Krylov State Spaces and \cref{theorem:E_d}}\label{app2:proof_Km_dim}
\setcounter{theorem}{0} 

\begin{theorem}[Repeated]
Let \( H \in \mathbb{C}^{N \times N} \) be a Hermitian Hamiltonian with \( d \) pairwise distinct eigenvalues \( \varepsilon_0, \varepsilon_1, \ldots, \varepsilon_{d-1} \), and let \( \{ \ket{\phi_j} \} \) be an orthonormal eigenbasis of \( H \), satisfying
\begin{align}
    H \ket{\phi_j} = \varepsilon_j \ket{\phi_j}.
\end{align}
Then the time-evolved state \( \ket{\Psi(t)} = e^{-iHt} \ket{\Psi_0} \) lies in a \( d \)-dimensional subspace \( \mathrm{E}_d \subseteq \mathbb{C}^N \), i.e.,
\begin{align}
    \ket{\Psi(t)} \in \mathrm{E}_d := \mathrm{Span} \left\{ \ket{\xi_0}, \ket{\xi_1}, \ldots, \ket{\xi_{d-1}} \right\},
\end{align}
where the vectors \( \ket{\xi_p} \) are defined by
\begin{align}
    \ket{\xi_p} := \frac{1}{\sqrt{|J_p|}} \sum_{j \in J_p} \alpha_j \ket{\phi_j}, ~~~ \mathrm{with~ } \alpha_j := \bra{\phi_j}\ket{\Psi_0},
\end{align}

and \( J_p := \{ j \mid \varepsilon_j = \varepsilon_p \} \) denotes the set of indices corresponding to the degenerate eigenspace of eigenvalue \( \varepsilon_p \). The normalization factor \( |J_p| \) is the cardinality of the set \( J_p \).
The tarting state \( \ket{\psi_0} \) can be represented in the basis \( \{\ket{\xi_p}\}_p \), with \( \gamma_p = \bra{\xi_p}\ket{\psi_0} \), as
\begin{align}
    \ket{\psi_0} = \sum_{p=0}^{d-1} \gamma_p \ket{\xi_p}.
\end{align}
Let \( n_1 \) denote the number of coefficients for which \( \gamma_p = 0 \). Then, the number of linearly independent vectors is reduced to \( d - n_1 \). It also holds that for the Krylov state space \( \mathrm{K}_m = \mathrm{Span}\{\ket{\psi_0}, H\ket{\psi_0}, \ldots, H^{m-1}\ket{\psi_0}\} \),
\begin{align}
    m = d - n_1
\end{align}
holds.

\end{theorem}

\begin{proof}
Starting with the initial state \( \ket{\Psi_0} \in \mathbb{C}^N \), expand it in the eigenbasis of \( H \):

\begin{align}
    \ket{\Psi_0} = \sum_j \alpha_j \ket{\phi_j}, \alpha_j = \bra{\phi_j} \ket{ \Psi_0}
\end{align}

The time-evolved state is then given by
\begin{align}
    \ket{\Psi(t)} = e^{-iHt} \ket{\Psi_0} = \sum_j e^{-i \varepsilon_j t} \alpha_j \ket{\phi_j}.
\end{align}
Since the eigenvalues \( \varepsilon_j \) take only \( d \) distinct values, we partition the index set \( \{0, 1, \ldots, N-1\} \) into disjoint subsets \( J_p \), where
\[
J_p := \{ j \in \{0, \ldots, N-1\} \mid \varepsilon_j = \varepsilon_p \}, \quad \mathrm{for ~} p = 0, \ldots, d-1.
\]
This lets us rewrite the time-evolved state as:
\begin{align}
    \ket{\Psi (t)} &= \sum_{j} e^{-i \varepsilon_j t} \alpha_j \ket{\phi_j}
    = \sum_{p=0}^{d-1} \sum_{j \in J_p} e^{-i \varepsilon_p t} \alpha_j \ket{\phi_j} \nonumber \\
    &= \sum_{p=0}^{d-1} e^{-i \varepsilon_p t} \left( \sum_{j \in J_p} \alpha_j \ket{\phi_j} \right).
\end{align}
Define the time-independent vectors
\begin{align}
    \ket{\xi_p} := \frac{1}{\sqrt{|J_p|}} \sum_{j \in J_p} \alpha_j \ket{\phi_j},
\end{align}
which represent normalized superpositions. Then, the time-evolved state becomes
\begin{align}
    &\ket{\Psi(t)} = \sum_{p=0}^{d-1} e^{-i \varepsilon_p t} \sqrt{|J_p|} \ket{\xi_p} \nonumber \\
    \Rightarrow &\ket{\Psi(t)} \in \mathrm{Span}\{\ket{\xi_0}, \ldots, \ket{\xi_{d-1}}\} = \mathrm{E}_d.
\end{align}
Thus, \( \ket{\Psi(t)} \) lies in the \( d \)-dimensional subspace \( \mathrm{E}_d \) for all \( t \in \mathbb{R} \).
Given the starting state \( \ket{\psi_0} \), it can be represented in the basis \( \{\ket{\xi_p}\}_p \) with \( \gamma_p = \bra{\xi_p}\ket{\psi_0} \), as
\begin{align}
    \ket{\psi_0} = \sum_{p=0}^{d-1} \gamma_p \ket{\xi_p}.
\end{align}
This is a superposition of \( d \) linearly independent vectors. Let \( n_1 \) denote the number of coefficients for which \( \gamma_p = 0 \). Then, the number of linearly independent vectors is reduced to \( d - n_1 \). Since \( \ket{\xi_p} \) is given by
\begin{align}
    \ket{\xi_p} = \sum_{j \in J_p} \alpha_j \ket{\phi_j},
\end{align}
it follows that \( \gamma_p = 0 \) only if all \( \alpha_j = \bra{\phi_j}\ket{\Psi_0} = 0 \) for all \( j \in J_p \). 
Since \cref{lemma:d=min} and \cref{lemma:m=min} show that both $\mathrm{K}_m$ and $\mathrm{E}_{d-n_1}$ consist of the smallest possible number of basis states, it follows the dimension of both must be equal, i.e. \( m = d - n_1 \). 
\end{proof}

\begin{lemma}[$\mathrm{E}_m$ consists of the minimum number of basis states.]\label{lemma:d=min}
Given an initial state $\ket{\Psi_0}$, a Hamiltonian $H$, and the corresponding space of eigenstates $\mathrm{E}_d$, such that $\ket{\Psi(t)} \in \mathrm{K}_m$ for all $t \in \mathbb{R}$,  
there exists no basis $\mathcal{B}$ with $\mathrm{dim}(\mathcal{B}) < d$ such that any time-evolved state is in the span of $\mathcal{B}$.
\begin{proof}
Assume times $t_0 < t_1 < \ldots < t_{d-1}$ and the states evolved at those times are given by 
\begin{align}
    \ket{\Psi(t_j)} &= e^{-iHt_j}\ket{\Psi_0} 
    = \sum_{k=0}^{N} e^{-i\varepsilon_k t_j} \ket{\phi_k} \bra{\phi_k} \ket{\Psi_0} \nonumber \\
    &= \sum_{p=0}^{d-1} e^{-i\varepsilon_p t_j} \sum_{j \in J_p} \alpha_j \ket{\phi_j} 
    = \sum_{p=0}^{d-1} e^{-i\varepsilon_p t_j} \ket{\xi_p} \nonumber \\
    &= 
    \begin{pmatrix}
        \ket{\xi_0} & \ket{\xi_1} & \ldots & \ket{\xi_{d-1}}
    \end{pmatrix}
    \begin{pmatrix}
        e^{-i\varepsilon_0 t_j} \\
        e^{-i\varepsilon_1 t_j} \\
        \vdots \\
        e^{-i\varepsilon_{d-1} t_j}
    \end{pmatrix}
\end{align}
Here, $\ket{\xi_p} = \sum_{j \in J_p} \alpha_j \ket{\phi_j}$ and $\alpha_j = \bra{\phi_j} \ket{\Psi_0}$ is used. We do not normalize $\ket{\xi_i}$, since it does not change anything but allows better readability. Writing the $d$ states in terms of the $\ket{\xi_p}$ basis results in
\begin{align}
\begin{pmatrix}
    \ket{\Psi(t_0)} &  \ldots & \ket{\Psi(t_{d-1})}
\end{pmatrix}
=
\begin{pmatrix}
    \ket{\xi_0} & \ldots & \ket{\xi_{d-1}}
\end{pmatrix}
\Sigma \nonumber \\
\Sigma = 
\begin{pmatrix}
    e^{-i\varepsilon_0 t_0} & e^{-i\varepsilon_0 t_1} & \ldots & e^{-i\varepsilon_0 t_{d-1}} \\
    e^{-i\varepsilon_1 t_0} & e^{-i\varepsilon_1 t_1} & \ldots & e^{-i\varepsilon_1 t_{d-1}} \\
    \vdots & \vdots & \ddots & \vdots \\
    e^{-i\varepsilon_{d-1} t_0} & e^{-i\varepsilon_{d-1} t_1} & \ldots & e^{-i\varepsilon_{d-1} t_{d-1}}
\end{pmatrix}
\label{eq:mat}
\end{align}

If the matrix $\Sigma$ is invertible, then it follows that any $d$ time-evolved states span $\mathrm{E}_d$. If $\abs{\varepsilon_j} < \pi$, we can write $x_j^t = (e^{-i\varepsilon_j})^t$. If this is not the case, the Hamiltonian can be rescaled as $H \leftarrow H / \abs{\varepsilon_{\max}}$ and time as $t \leftarrow t \cdot \abs{\varepsilon_{\max}}$, which results in the same dynamics. With this, the matrix in \cref{eq:mat} can be rewritten as
\begin{align}
\begin{pmatrix}
    \ket{\Psi(t_0)} &  \ldots & \ket{\Psi(t_{d-1})}
\end{pmatrix}
=
\begin{pmatrix}
    \ket{\xi_0} &  \ldots & \ket{\xi_{d-1}}
\end{pmatrix}
\Sigma
\end{align}
where
\begin{align}
\Sigma = \begin{pmatrix}
    x_0^{t_0} & x_0^{t_1} & \ldots & x_0^{t_{d-1}} \\
    x_1^{t_0} & x_1^{t_1} & \ldots & x_1^{t_{d-1}} \\
    \vdots & \vdots & \ddots & \vdots \\
    x_{d-1}^{t_0} & x_{d-1}^{t_1} & \ldots & x_{d-1}^{t_{d-1}}
\end{pmatrix}
\end{align}

For pairwise distinct times, the columns are linearly independent if $t_{d-1} < T_P$, where $T_P$ is the period of the system. $\Sigma$ is a generalized Vandermonde matrix, for which an inverse exists. Therefore,
\begin{align}
\begin{pmatrix}
    \ket{\Psi(t_0)} & \ldots & \ket{\Psi(t_{d-1})}
\end{pmatrix}
\Sigma^{-1}
=
\begin{pmatrix}
    \ket{\xi_0} &  \ldots & \ket{\xi_{d-1}}
\end{pmatrix}
\end{align}

Since all time-evolved states lie in $\mathrm{E}_d$, and since these $d$ time-evolved states span the space, it follows that $\mathrm{E}_d$ consists of the minimum number of basis vectors.
\end{proof}
\end{lemma}

\begin{lemma}[$\mathrm{K}_m$ consists of the minimum number of basis states.]\label{lemma:m=min}
Assume $\ket{\Psi(t)}$ is a time-evolved state under Hamiltonian $H$, and assume that the corresponding Krylov space is given by $\mathrm{K}_m$, such that $\ket{\Psi(t)} \in \mathrm{K}_m$ for all $t \in \mathbb{R}$. Then it holds that there exists no basis $\mathcal{B}$ with $\mathrm{dim}(\mathcal{B}) < m$ such that $\ket{\Psi(t)} \in \mathrm{Span}\{\mathcal{B}\}$.
\begin{proof}
    The proof proceeds as follows. Take $L > m-1$ time-evolved states $\ket{\Psi(t_j)} = e^{-iHt_j} \ket{\Psi_0}$ at times $0 = t_0 < t_1 < t_2 < \ldots < t_L < T_P$. Pick $L$ very large, such that
    \begin{align}
        \ket{\Psi(t_j)} = \sum_{k=0}^{\infty} (-iH)^k \frac{t_j^k}{k!} \ket{\Psi_0} = \sum_{k=0}^{L} (-iH)^k \frac{t_j^k}{k!} \ket{\Psi_0} + \varepsilon,
    \end{align}
    holds with $\varepsilon\rightarrow 0$ for $L\rightarrow \infty $. Introduce the substitution:
    \begin{align}
        h_j(\ket{\Psi_0}) &:= \sum_{k=0}^{L} f^k(\ket{\Psi_0}) \frac{t_j^k}{k!}, \\
        \text{with } f^k(\ket{\Psi_0}) &:= (-iH)^k \ket{\Psi_0}. \nonumber
    \end{align}
    For readability, the dependence on \( \ket{\Psi_0} \) is ignored in the following calculations. The vectors \( h_j \) can be written as
    \begin{align}
        h_i = \sum_{j=0}^{n} f^j \frac{t_i^j}{j!} 
        = 
        \begin{pmatrix}
            f^0 & f^1 & \ldots & f^{n}    
        \end{pmatrix}
        \begin{pmatrix}
            1\\
            t_i / 1!\\
            t_i^2 / 2! \\
            \vdots \\
            t_i^{n} / {n}!
        \end{pmatrix}.
    \end{align}
    Writing all vectors $h_0, \ldots, h_{n}$ yields
    \begin{align}
         \begin{pmatrix}
             h_0 & h_1 & \ldots & h_{n}
         \end{pmatrix}
         = 
         \begin{pmatrix}
            f^0 & f^1 & \ldots & f^{n}    
         \end{pmatrix}
         \Theta, \nonumber \\
         \Theta = \begin{pmatrix}
                1 & 1 & \ldots & 1 \\
                t_1 / 1! & t_2 / 1! & \ldots & t_{n} / 1! \\
                t_1^2 / 2! & t_2^2 / 2! & \ldots & t_{n}^2 / 2! \\
                \vdots & \vdots & \ddots & \vdots \\
                t_1^{n} / {n}! & t_2^{n} / {n}! & \ldots & t_{n}^{n} / {n}!
            \end{pmatrix}.
    \end{align}
    Since all times $t_i$ are pairwise distinct by definition, the columns of $\Theta$ are linearly independent. Therefore, $\Theta$ is invertible, and its inverse $\Theta^{-1}$ exists:
    \begin{align}
        \begin{pmatrix}
             h_0 & h_1 & \ldots & h_{n}
         \end{pmatrix} \Theta^{-1}
         = 
         \begin{pmatrix}
            f^0 & f^1 & \ldots & f^{n}    
         \end{pmatrix}.
    \end{align}
    Since all vectors $\{f^0, f^1, \ldots, f^{L}\}$ can be represented as linear combinations of vectors $\{h_0, h_1, \ldots, h_{L}\}$, the spans of both sides are equal:
    \begin{align}
        H_L^L &= \mathrm{Span}\{h_0, h_1, \ldots, h_{n}\} \nonumber \\
         &= \mathrm{Span}\{f^0, f^1, \ldots, f^{L}\} = K_L. \label{eq:help1}
    \end{align}
    Since $K_L$ contains only $m$ linearly independent vectors, $H_L^L$ must also contain only $m$ linearly independent vectors ($H_m^L = K_m$). This implies that there cannot exist a basis $\mathcal{B}$ with $\mathrm{dim}(\mathcal{B}) < m$ such that $\ket{\Psi(t)} \in \mathrm{Span}\{\mathcal{B}\}$. If such a basis existed, then the $L$ vectors $h_j$ would not be able to represent the vectors $f^j$, and the matrix $\Theta$ would not be invertible, which would be contradicting the fact that $\Theta$ is indeed invertible.
\end{proof}
\end{lemma}

\section{Proofs for Operator Spaces and  \cref{theorem:L_M_grade}}\label{app3:proof_L_M_grade}
\setcounter{theorem}{1} 

\begin{theorem}[Repeated]
\label{theorem:L_M_grade}

Let \( H \in \mathbb{C}^{N \times N} \) be a Hamiltonian with eigenbasis \( \{\ket{\phi_j}\} \) and corresponding eigenvalues \( \varepsilon_j \), and let \( O \) be an operator on the same Hilbert space. Define the Liouvillian Krylov space
\[
\mathcal{L}_M := \mathrm{Span}\{\mathcal{L}^0(O), \mathcal{L}^1(O), \ldots, \mathcal{L}^{M-1}(O)\},
\]
where \( \mathcal{L}(O) = [H, O] \) is the Liouvillian superoperator.

Define the transition frequencies \( \omega_{mn} := \varepsilon_m - \varepsilon_n \), and let \( \{\omega_P\}_{P=0}^{N_{\omega}-1} \) be the set of all pairwise distinct values taken by \( \omega_{mn} \). For each \( \omega_P \), define the index set
\[
J_P := \left\{ (m,n) \,\middle|\, \omega_{mn} = \omega_P \right\},
\]
and the corresponding matrix
\[
\sigma_P := \sum_{(m,n) \in J_P} \bra{\phi_m} O \ket{\phi_n} \ket{\phi_m} \bra{\phi_n}.
\]
Let \( N_1 \) is the number of vanishing contributions \( \sigma_P = 0 \), then the time-evolved operator is given by
\[
O(t) = \sum_{P \in S} e^{i\omega_P t} \sigma_P,
\]
with \( S = \{P \mid \sigma_P \neq 0\} = \{s_0, s_1, \ldots, s_{N_{\omega} - N_1 - 1}\} \). The operator lies in the span
\[
O(t) \in \mathrm{Span}\{\sigma_{s_0}, \sigma_{s_1}, \ldots, \sigma_{s_{N_{\omega} - N_1 - 1}}\}=\mathcal{P}_{N_\omega - N_1}.
\]
Further, the grade \( M \) of the Krylov space \( \mathcal{L}_M \) is given by
\[
M = N_{\omega} - N_1.
\]
\end{theorem}

\begin{proof}
The time-evolved operator is given by
\[
O(t) = e^{iHt} O e^{-iHt} = \sum_{m,n} e^{i(\varepsilon_m - \varepsilon_n)t} O_{mn} \ket{\phi_m} \bra{\phi_n},
\]
with matrix elements \( O_{mn} = \bra{\phi_m} O \ket{\phi_n} \).
Define the transition frequencies \( \omega_{mn} := \varepsilon_m - \varepsilon_n \), and let \( \{\omega_P\}_{P=0}^{N_{\omega}-1} \) be the set of all pairwise distinct values taken by \( \omega_{mn} \). For each \( \omega_P \), define the index set
\[
J_P := \{(m,n) \mid \omega_{mn} = \omega_P\}.
\]
Then the operator can be grouped as
\begin{align*}
O(t) &= \sum_{P=0}^{N_{\omega}-1} e^{i\omega_P t} \sum_{(m,n) \in J_P} O_{mn} \ket{\phi_m} \bra{\phi_n} \nonumber \\
&= \sum_{P=0}^{N_{\omega}-1} e^{i\omega_P t} \sigma_P,
\end{align*}
where we define
\[
\sigma_P := \sum_{(m,n) \in J_P} O_{mn} \ket{\phi_m} \bra{\phi_n}.
\]
\textit{The functions \( e^{i \omega_P t} \) are linearly independent over \( \mathbb{R} \) when the frequencies \( \omega_P \) are pairwise distinct.}
Due to the Hermiticity of \( H \), the vectors \( \ket{\phi_m} \) are linearly independent, which implies the linear independence of \( \sigma_P \). Let \( N_1 \) denote the number of such zero contributions, i.e. \( \sigma_P = 0 \). Then the sum reduces to
\[
O(t) = \sum_{P \in S} e^{i\omega_P t} \sigma_P,
\]
where \( S = \{P \mid \sigma_P \neq 0\} = \{s_0, s_1, \ldots, s_{N_{\omega} - N_1 - 1}\} \), and the operator lies in the span
\[
O(t) \in \mathrm{Span}\{\sigma_{s_0}, \sigma_{s_1}, \ldots, \sigma_{s_{N_{\omega} - N_1 - 1}}\}=\mathcal{P}_{N_\omega - N_1}.
\]
By \cref{lemma:Pm_minimal} it holds that this space is minimal, i.e, there exists no basis \( \mathcal{B} \) with \( \dim(\mathcal{B}) < N_{\omega} - N_1 \) such that \( O(t) \in \mathrm{Span}(\mathcal{B}) \).
On the other hand, the time-evolved operator lies in the Liouvillian Krylov space,
\[
O(t) \in \mathcal{L}_M = \mathrm{Span}\{\mathcal{L}^0(O), \mathcal{L}^1(O), \ldots, \mathcal{L}^{M-1}(O)\}.
\]
Since due to \cref{lemma:L=min} it holds $\mathcal{L}^{M-1}(O)\}$ is a minimal space as well, we conclude that
\[
M = N_{\omega} - N_1.
\]

\end{proof}

\begin{lemma}[$\mathcal{P}_{N_\omega - N_1}$ consists of the minimum number of basis operators.]\label{lemma:Pm_minimal}
Given an operator \( O \in \mathbb{C}^{N \times N} \), a Hamiltonian \( H \) with eigenvalues \( \varepsilon_j \), and the associated decomposition of the time-evolved operator \( O(t) = e^{iHt} O e^{-iHt} \), let \( \mathcal{P}_{N_\omega - N_1} = \mathrm{Span}\{\sigma_0, \ldots, \sigma_{N_\omega - N_1 - 1}\} \) denote the space generated by the \( N_\omega - N_1 \) non-zero frequency components in the Liouvillian decomposition as discussed in \cref{theorem:L_M_grade}. Then there exists no smaller set of operators that spans all \( O(t) \) for arbitrary time \( t \in \mathbb{R} \).
\begin{proof}
Assume times \( t_0 < t_1 < \ldots < t_{N_\omega - N_1 - 1} \), and write the time-evolved operators in the basis discussed in \cref{theorem:L_M_grade} at those times:
\begin{align}
    O(t_j) &= \sum_{P=0}^{N_\omega - N_1 - 1} e^{i \omega_P t_j} \sigma_P \nonumber \\
    &= \begin{pmatrix}
        \sigma_0 & \sigma_1 & \ldots & \sigma_{N_\omega - N_1 - 1}
    \end{pmatrix}
    \begin{pmatrix}
        e^{i\omega_0 t_j} \\
        e^{i\omega_1 t_j} \\
        \vdots \\
        e^{i\omega_{N_\omega - N_1 - 1} t_j}
    \end{pmatrix}
\end{align}
Stacking the \( N_\omega - N_1 \) operators into a matrix gives
\begin{align}
\begin{pmatrix}
    O(t_0) & \ldots & O(t_{N_\omega - N_1 - 1})
\end{pmatrix}
=
\begin{pmatrix}
    \sigma_0 & \ldots & \sigma_{N_\omega - N_1 - 1}
\end{pmatrix}
\Sigma
\end{align}
where
\begin{align}
\Sigma = \begin{pmatrix}
    e^{i \omega_0 t_0} &  \ldots & e^{i \omega_0 t_{N_\omega - N_1 - 1}} \\
    e^{i \omega_1 t_0} &  \ldots & e^{i \omega_1 t_{N_\omega - N_1 - 1}} \\
    \vdots& \ddots & \vdots \\
    e^{i \omega_{N_\omega - N_1 - 1} t_0} & \ldots & e^{i \omega_{N_\omega - N_1 - 1} t_{N_\omega - N_1 - 1}}
\end{pmatrix}
\end{align}
This matrix \( \Sigma \) is a generalized Vandermonde matrix in the variables \( x_P := e^{i \omega_P} \). If the frequencies \( \omega_P \) are distinct and the time points \( t_j \) are also distinct (and less than the system period \( T_P \)), then \( \Sigma \) is invertible. This holds due to the construction of $\mathcal{P}_{N_\omega - N_1}$ in \cref{theorem:L_M_grade}.

Therefore, we can recover the \( \sigma_P \) as:
\begin{align}
\begin{pmatrix}
    \sigma_0 & \ldots & \sigma_{N_\omega - N_1 - 1}
\end{pmatrix}
=
\begin{pmatrix}
    O(t_0) & \ldots & O(t_{N_\omega - N_1 - 1})
\end{pmatrix}
\Sigma^{-1}
\end{align}

Hence, the \( N_\omega - N_1 \) time-evolved operators span the space \( \mathcal{P}_{N_\omega - N_1} \), and no smaller set of operators can generate the time evolution. Therefore, \( \mathcal{P}_{N_\omega - N_1} \) is minimal.
\end{proof}
\end{lemma}

\begin{lemma}[$\mathcal{L}_M$ Consists of the Minimum Number of Basis States]
\label{lemma:L=min}
Let $O$ be an operator and $H$ the Hamiltonian under which $O$ evolves. Then, there exists no space $W$ such that $O(t) \in W$ for all $t$ and $\mathrm{dim}(W) < \mathrm{dim}(\mathcal{L}_M)$.
\begin{proof}
Let $0 = t_0 < t_1 < \ldots < t_Q$ be a time discretization with $t_Q < T_P$, where $T_P$ is the period of the system. Define:
\begin{align*}
    \tilde{O}(t_a) = \sum_{k=0}^Q \frac{(it_a)^k}{k!} \mathcal{L}^k(O)
\end{align*}
as an approximation of the time-evolved operator at time $t_a$. Then,
\begin{align*}
    \tilde{O}(t_a) = 
    (i^0 \mathcal{L}^0(O), i^1 \mathcal{L}^1(O), \ldots, i^Q \mathcal{L}^Q(O)) 
    \begin{pmatrix}
        1\\
        t_a\\
        \vdots\\
        {t_a^Q}/{Q!}
    \end{pmatrix}.
\end{align*}

Writing all $\tilde{O}(t_a)$ with pairwise distinct $t_a$ gives:
\begin{align}
    (\tilde{O}(t_0), \ldots, \tilde{O}(t_Q)) = 
    (i^0 \mathcal{L}^0(O), \ldots, i^Q \mathcal{L}^Q(O)) \Theta,
\end{align}
with
\begin{align}
    \Theta = \begin{pmatrix}
        1 & 1 & \ldots & 1 \\
        t_0 & t_1 & \ldots & t_Q \\
        t_0^2/2! & t_1^2/2! & \ldots & t_Q^2/2! \\
        \vdots & \vdots & \ddots & \vdots \\
        t_0^Q/Q! & t_1^Q/Q! & \ldots & t_Q^Q/Q!
    \end{pmatrix}.
\end{align}
This is the same representation as in \cref{lemma:m=min}, where $\Theta$ is a generalized Vandermonde matrix and is thus invertible. Therefore:
\begin{align}
    (\tilde{O}(t_0), \ldots, \tilde{O}(t_Q)) \Theta^{-1} = 
    (i^0 \mathcal{L}^0(O), \ldots, i^Q \mathcal{L}^Q(O)).
\end{align}

This shows that the time-evolved operators can reconstruct the powers of the Liouvillian and since for all time-evolved operators $O(t)\in \mathcal{L}_M$ holds, and each Liouvillian power can be expressed using the time-evolved operators, it follows that $\mathcal{L}_M$ has the minimum possible number of basis elements.
\end{proof}
\end{lemma}

\section{Krylov space dimension for the Hamiltonians}\label{app4:KrylovSpaceDim}
\begin{table}[H]
\centering
\begin{tabular}{llccccc}
\toprule
\textbf{Hamiltonian} & \textbf{Observable} & $~d~$ & $~d^2~$ & $~N_{\omega}~$ & $~N_1~$ & $~M~$ \\
\midrule
$H_{I1}$ & $Z_1$ & 9  & 81  & 71  & 40  & 31 \\
$H_{I1}$ & $Z_2$ & 9  & 81  & 71  & 40  & 31 \\
$H_{I1}$ & $Z_3$ & 9  & 81  & 71  & 40  & 31 \\
$H_{I1}$ & $Z_4$ & 9  & 81  & 71  & 40  & 31 \\
\midrule
$H_{I2}$ & $Z_1$ & 16 & 256 & 237 & 176 & 61 \\
$H_{I2}$ & $Z_2$ & 16 & 256 & 237 & 176 & 61 \\
$H_{I2}$ & $Z_3$ & 16 & 256 & 237 & 158 & 79 \\
$H_{I2}$ & $Z_4$ & 16 & 256 & 237 & 158 & 79 \\
\midrule
$H_{I3}$ & $Z_1$ & 15 & 225 & 211 & 112 & 99 \\
$H_{I3}$ & $Z_2$ & 15 & 225 & 211 & 112 & 99 \\
$H_{I3}$ & $Z_3$ & 15 & 225 & 211 & 112 & 99 \\
$H_{I3}$ & $Z_4$ & 15 & 225 & 211 & 112 & 99 \\
\midrule
$H_{I4}$ & $Z_1$ & 16 & 256 & 241 & 128 & 113 \\
$H_{I4}$ & $Z_2$ & 16 & 256 & 241 & 128 & 113 \\
$H_{I4}$ & $Z_3$ & 16 & 256 & 241 & 128 & 113 \\
$H_{I4}$ & $Z_4$ & 16 & 256 & 241 & 128 & 113 \\
\bottomrule
\end{tabular}
\caption{Spectral and dynamical statistics for each (Hamiltonian, Observable) pair. $d$ is the number of pairwise distinct eigenvalues, $d^2$ is the number of eigenvalue pairs, $n_{\omega}$ the number of distinct transition frequencies, $N_1$ the number of zero $\sigma_P$ contributions, and $M$ the resulting Krylov dimension.}
\label{tab:spectral_summary_all}
\end{table}
\clearpage

\bibliography{lit}

\end{document}